\documentclass[10pt, final]{iopart}

\makeatletter
\long\def\@makefntext#1{\parindent 1em\noindent
        \hb@xt@1.8em{%
            \hss\@textsuperscript{\normalfont\@thefnmark}}#1}%
\def\@makefnmark{\hbox{$^{\arabic{footnote}}\m@th$}}
\def\@thefnmark{\arabic{footnote}}
\makeatother
\expandafter\let\csname equation*\endcsname\relax
\expandafter\let\csname endequation*\endcsname\relax

\usepackage{amsmath, amsfonts, amssymb, amsthm}
\usepackage{physics}
\usepackage{braket}
\usepackage{enumitem}
\usepackage[mode=buildnew]{standalone}
\usepackage{tikz}

\usepackage{soul}

\newtheorem{mydef}{Definition}
\newtheorem{theo}[mydef]{Theorem}
\newtheorem{lem}[mydef]{Lemma}

\newtheorem{corol}[mydef]{Corollary}

\usepackage{scalerel}
\usetikzlibrary{svg.path}

\definecolor{orcidlogocol}{HTML}{A6CE39}
\tikzset{
  orcidlogo/.pic={
    \fill[orcidlogocol] svg{M256,128c0,70.7-57.3,128-128,128C57.3,256,0,198.7,0,128C0,57.3,57.3,0,128,0C198.7,0,256,57.3,256,128z};
    \fill[white] svg{M86.3,186.2H70.9V79.1h15.4v48.4V186.2z}
                 svg{M108.9,79.1h41.6c39.6,0,57,28.3,57,53.6c0,27.5-21.5,53.6-56.8,53.6h-41.8V79.1z M124.3,172.4h24.5c34.9,0,42.9-26.5,42.9-39.7c0-21.5-13.7-39.7-43.7-39.7h-23.7V172.4z}
                 svg{M88.7,56.8c0,5.5-4.5,10.1-10.1,10.1c-5.6,0-10.1-4.6-10.1-10.1c0-5.6,4.5-10.1,10.1-10.1C84.2,46.7,88.7,51.3,88.7,56.8z};
  }
}

\newcommand\orcidicon[1]{\href{https://orcid.org/#1}{\mbox{\scalerel*{
\begin{tikzpicture}[yscale=-1,transform shape]
\pic{orcidlogo};
\end{tikzpicture}
}{|}}}}
\usepackage{hyperref}

\DeclareMathOperator{\ch}{ch}
\DeclareMathOperator{\openone}{1\!\!\!\!1}

\usepackage[numbers]{natbib}

\begin{document}
\title{Weakly coupled local particle detectors cannot harvest entanglement}
\author{Maximilian H. Ruep \orcidicon{0000-0001-6866-4506}}
\address{Department of Mathematics, University of York, Heslington, York YO10 5DD, United Kingdom
}
\ead{maximilian.ruep@york.ac.uk}
\vspace{10pt}
\begin{indented}
\item \today
\end{indented}

\begin{abstract}
Many states of linear real scalar quantum fields (in particular Reeh-Schlieder states) on flat as well as curved spacetime are entangled on spacelike separated local algebras of observables. It has been argued that this entanglement can be ``harvested'' by a pair of so-called particle detectors, for example \emph{singularly} or \emph{non-locally} coupled quantum mechanical harmonic oscillator Unruh detectors. In an attempt to avoid such imperfect coupling, we analyse a model-independent local and covariant entanglement harvesting protocol based on the local probes of a recently proposed measurement theory of quantum fields. We then introduce the notion of a local particle detector concretely given by a local mode of a linear real scalar probe field on possibly curved spacetime and possibly under the influence of external fields. In a non-perturbative analysis we find that local particle detectors cannot harvest entanglement below a critical coupling strength when the corresponding probe fields are initially prepared in quasi-free Reeh-Schlieder states and are coupled to a system field prepared in a quasi-free state. This is a consequence of the fact that Reeh-Schlieder states restrict to truly mixed states on any local mode.
\end{abstract}

\vspace{2pc}

\section{Introduction}

Entanglement is an intrinsic feature of relativistic quantum theory. It is very well known that the class of Reeh-Schlieder states (see Sec.~\ref{sec_subsec_Reeh-Schlieder_mixed} for a precise definition) of algebraic quantum field theories on possibly curved spacetime are entangled on {any} two spacelike separated local observable algebras{~\cite{VERCH_2005,halvorson_clifton200generic}}. This applies in particular to the vacuum of Klein-Gordon fields in Minkowski spacetime~\cite{VERCH_2005, SUMMERS1985257,summers1987, Summers1987I,Summers1987II, summersAIHPA_1988}, and similarly to any other states of bounded energy~\cite{haag2012local}. It has been argued that the entanglement of a quantum field (system) can be accessed via an \emph{entanglement harvesting} protocol, in which two agents couple two ``physical structures'' (probes), initially in an uncorrelated product state, to a quantum field in two spacelike separated regions and end up in an entangled final state{~\cite{VALENTINI1991, Reznik2003}}. {(See also~\cite{Benatti2004} for a different scenario, in which the internal degrees of freedom of a single probe get entangled.)}

Entanglement harvesting has been demonstrated theoretically {on} flat as well as curved spacetime where the ``physical structures'' are explicitly given by either a pair of two-level systems/qubits or a pair of quantum mechanical harmonic oscillators (Unruh detectors)\footnote{In~\cite{Unruh1976}, Unruh introduced (in addition to a relativistic model) a non-relativistic particle in a box as a particle detector. DeWitt took up this idea and discussed general point-like detector models with discrete internal energy states with a monopole type coupling in~\cite{DeWitt1979}. Unruh and Wald discussed a concrete realisation by a two-level system in~\cite{UnruhWald1984}.} coupled to a Klein-Gordon field, see for instance~\cite{Brown2013, Pozas-Kerstjens_2015, Henderson_2018}. Each of these non-relativistic quantum mechanical particle detectors is typically initially prepared in its \emph{pure} ground state. The final state of a pair of them is then analysed in perturbation theory and, at leading order, is found to be entangled \emph{independent of the coupling strength}. However, it is very well known that the intrinsic non-relativistic nature of the above detector models results in a \emph{singular} coupling \emph{on} or a \emph{non-local} coupling \emph{around} \emph{a} worldline: the underlying (classical) equation of motion of the coupled detector-quantum field structure is either singular or non-local~\cite{fewster-non-local}. Hence we refer to them as \emph{non-local particle detectors}. Despite the fact that this non-locality in the latter case can be controlled~\cite{Mart_n_Mart_nez_2015, deramon2021relativistic}, it is clear that non-local particle detectors can only function as either purely mathematical/technical tools or as non-relativistic approximations with limited area of applicability, unless one is willing to give up on the principle of locality (or ready to deal with singular equations). We also emphasise that the Unruh effect as well as Hawking radiation can both be derived without utilising non-local particle detectors, see in particular~Sec.~II.~in~\cite{FredenhagenHaag1987generally} and~\cite{ FredenhagenHaag1990hawking}.

The purpose of our paper is to circumvent such imperfections by applying the idea of entanglement harvesting to \emph{local probes} in the sense of~\cite{fewster2018quantum}. We show that this allows for a model-independent and local analysis of the corresponding protocol on flat as well as curved spacetime. The system as well as the two probes of the two agents ($\mathsf{A}, \mathsf{B}$) are all modelled by local quantum field theories on a globally hyperbolic spacetime $M$. The interaction between probe $\mathsf{J} \in \{\mathsf{A},\mathsf{B}\}$ and the system is assumed to be local and engineered in such a way that it is only active in a compact \emph{coupling zone} $K_\mathsf{J} \subseteq M$. It is unreasonable to assume that the agents have access to \emph{all} degrees of freedom of their respective local probe ``after'' the interaction with the system (in general this would either require infinite spatial extension or infinite amount of time). Hence we associate to agent $\mathsf{J}$ a ``processing region'' $N_\mathsf{J} \subseteq M$, characterising the spatial and temporal extension of agent $\mathsf{J}$'s interaction with the probe. This processing could for instance be a measurement of a probe observable localisable in $N_\mathsf{J}$ but may also be a more general interaction. In this model-independent formulation, entanglement harvesting corresponds to the entanglement of the final state of the two probes restricted to the combination of $\mathsf{A}$'s local observable algebra of $N_\mathsf{A}$ and $\mathsf{B}$'s local observable algebra of $N_\mathsf{B}$.

The use of non-local particle detectors follows in principle the same sentiment but differs in a crucial way from the \emph{local} protocol: the internal dynamics of, e.g., the Unruh detector is characterised by the flow of time only, i.e., it can be seen as a $0+1$-dimensional theory on $\mathbb{R}$ as opposed to that of a local probe, which lives (just as the system) on $M$. In this context the processing region can be seen as an interval of (or even a point in) time $\mathbb{R}$ \emph{after} the interaction. In particular, independent of the choice of processing region, the agent has always access to all probe degrees of freedom (in the case of the Unruh detector only one single mode). This is an artefact of the simplified internal dynamics of this detector model, which is completely decoupled from the spacetime picture.

We further discuss an explicit probe model given by a linear real scalar field (possibly under the influence of external fields). We introduce the notion of a \emph{local} particle detector by restricting our attention to certain local observables, i.e., to a single mode of the whole field localisable in the chosen processing region of spacetime. By choosing a local mode we single out one degree of freedom, a harmonic oscillator immersed in the field, that we deem accessible to an experimenter. In general, due to the locality of the field, the accessible local observables strongly depend on the choice of processing region. Heuristically, the spatial extent of different spacetime localisation regions of a local observable expands, which poses a challenge for experimenters: minimising dispersive effects throughout preparation of the probe, interaction with the system and final processing requires careful experimental design. By choosing a general linear real scalar field, we are equipped with the possibility to include external fields (for instance, a spacetime dependent mass), which can be used to model an experimental setup such as a cavity or a trap. We briefly sketch this idea, however, a detailed analysis is beyond the scope of our paper.

Having defined a local particle detector, we can now ask the question how to characterise reasonable initial preparation states. We show that whenever the probe field is initially prepared in a Reeh-Schlieder state, then the state of any local mode is truly \emph{mixed} (in particular not its ground state). For instance ground and KMS (thermal) states of Klein-Gordon fields on stationary spacetimes have the Reeh-Schlieder property~\cite{Strohmaier2000} and {moreover can be taken to be} quasi-free (Gaussian)\footnote{In fact, for Klein-Gordon fields on stationary spacetimes it holds that $C^2$-regular ground and extremal $C^1$-regular KMS states with vanishing one-point function are quasi-free (with vanishing one-point function), see~\cite{sanders2013thermal} for the results and the relevant definitions.}. We then consider the entanglement harvesting protocol for a linear real scalar system field prepared in a quasi-free state and two bilinearly coupled linear real scalar probe fields each prepared in a quasi-free Reeh-Schlieder state. The final state of the two local particle detectors is then quasi-free as well and its entanglement can be fully analysed~\cite{Simon_2000}. We find that in the described scenario local particle detectors cannot harvest entanglement at arbitrarily small coupling: for fixed local particle detectors and fixed preparation states of probes and system there exists a critical coupling strength below which entanglement harvesting is impossible\footnote{This is not to be confused with the ``entanglement death-zone'' of for instance~\cite{Cong_2019}, where no entanglement can be harvested by non-local particle detectors near moving mirrors. The coupling strength threshold for entanglement harvesting  we report on here is in the coupling parameter space.}. This is the main result of our paper.

For the convenience of the reader we give an outline of the following text.

In Sec.~\ref{Sec_local_probes} we first give an introduction to the FV framework (named after the authors of~\cite{fewster2018quantum}), its local probes and the notion of entanglement in the algebraic approach to physics. 

In Sec.~\ref{Sec_Model_indep_E_H} we show that local probes are an effective tool to formulate the entanglement harvesting protocol in a local and covariant manner. We consider two agents $\mathsf{A}, \mathsf{B}$ and rigorously establish necessary conditions on the relative causal relationship between the compact coupling zones $K_\mathsf{A}, K_\mathsf{B}$ and the processing regions $N_\mathsf{A}, N_\mathsf{B}$ in order to harvest any correlation. In particular, in the case of spacelike separated coupling zones we show that in order to harvest classical correlation [entanglement], the initially prepared system state must be classically correlated [entangled] on the combination of system observables localisable in (a connected region containing) $K_\mathsf{A}$  with the system observables localisable in (a connected region containing) $K_\mathsf{B}$. These reasonable results show that entanglement harvesting is no ``spooky action at a distance''.

In Sec.~\ref{Sec_local_particle_detectors} we consider the covariantly quantised linear real scalar field on possibly curved spacetime and possibly under the influence of external fields as a local probe and show how the restriction to one \emph{local} mode (in the processing region) can be viewed as a \emph{local} particle detector. We crucially demonstrate that every Reeh-Schlieder state restricts to a truly mixed state on any local mode. In particular, our interpretation is that local particle detectors cannot be (physically reasonably) prepared in their ground state resulting in the presence of underlying noise in any measured signal.

Sec.~\ref{Sec_E_H_local} consists of the analysis of the entanglement harvesting protocol for local particle detectors, whose associated probe field is bilinearly coupled to a linear real scalar quantum (system) field. In the case of quasi-free preparation states of system and probes the final state of the two local particle detectors is quasi-free and can be fully analysed~\cite{Vidal_2002, Simon_2000}. It turns out that the purity of the initial probe states when restricted to the local modes plays an important role. In the case in which the initial states of the local particle detectors are \emph{not} pure, a non-perturbative analysis shows that for fixed modes and fixed initial probe and system states there exists a coupling strength threshold for entanglement harvesting. Below a certain critical coupling strength no entanglement can be harvested. This is in particular the case for {initial} quasi-free Reeh-Schlieder probe states. We also indicate how a further perturbative analysis can be performed, which might give insight into the magnitude of the critical coupling before we give an outlook and conclude in Sec.~\ref{Sec_outlook}.

\section{Local probes}
\label{Sec_local_probes}

This section serves as a minimal introduction to the previously mentioned FV framework of local probes. The framework was introduced in~\cite{fewster2018quantum}, see also~\cite{fewster2019generally} for a summary and~\cite{bostelmann2020impossible} for a heuristic overview and a concise presentation. 

The reason behind using local probes is the same operationally motivated idea as for non-local particle detectors: any measurement of, or even any interaction with a physical system of interest (e.g.\ a quantum field) is performed via coupling a (measurement) device ``for a certain period of time'' to the system. Afterwards the device itself may be subject to further investigation or processing (such as reading off a pointer value, etc.) and is then discarded (traced out)\footnote{In the entanglement harvesting protocol the focus lies on the state of the two probes after the interaction, so we will naturally discard the system rather than the probes.}. The crucial point is that in the FV framework, the measurement device is given by a \emph{bona fide} local probe theory, which is coupled to the system of interest (e.g. the local quantum field) in a \emph{compact} coupling zone of spacetime in a \emph{local} manner. This idea has been formalised and cast in the general language of algebraic quantum field theory (AQFT) on possibly curved spacetime (or likewise in the presence of external fields) in~\cite{fewster2018quantum}. The emerging framework has been successfully employed in~\cite{bostelmann2020impossible}, where it was shown to implement multiple successive measurements of (or more generally operations on) quantum fields in a fully covariant and causal way, thereby avoiding superluminal signalling issues as raised in~\cite{sorkin1993impossible}.

For the reader's convenience we now quickly recall some crucial notions of Lorentzian geometry as well as AQFT with a focus on bipartite systems and entanglement of states on them. 

\subsection{AQFT {on} curved spacetime}

\subsubsection{Lorentzian geometry} 

A globally hyperbolic spacetime $M$ is a time-oriented Lorentzian spacetime of dimension greater or equal two that contains a Cauchy surface. We fix a globally hyperbolic $M$; for $N \subseteq M$ we denote by $J^+(N)$ and $J^-(N)$ its causal future and past in $M$ respectively and by $\ch(N):= J^+(N) \cap J^-(N)$ its causal hull. $N$ is called causally convex (in $M$) if $N=\ch(N)$. Non-empty open causally convex subsets of $M$ are called \emph{regions} and are globally hyperbolic spacetimes in their own right. The causal complement of a subset $K$ is $K^\perp := M \setminus (J^+(K) \cup J^-(K))$. For a compact subset $K$, the sets $M \setminus J^\mp(K)$ and $K^\perp$ are either empty or regions, see, for example, the Appendix~A of~\cite{fewster2012dynamical} for details and proofs.

\subsubsection{Algebraic quantum field theory}

In the algebraic model-independent approach to physics, the collection of observables of a physical system (classical or quantum) has the structure of a unital $*$-algebra $\mathcal{A}$. Concrete realisations are, e.g., complex-valued functions on phase space in classical mechanics, complex-valued functionals of field-configurations in classical field theory or operators on a Hilbert space in quantum mechanics. In AQFT, this idea is paired with the notion of \emph{locality}. (Hence AQFT is also known as \emph{local quantum physics} (LQP)~\cite{haag2012local}.) 

For a globally hyperbolic spacetime $M$ an algebraic quantum field theory (AQFT)\footnote{This convention also allows classical structures, i.e. Abelian algebras to constitute an AQFT.}, or simply a \emph{theory}, on $M$ consists of a $*$-algebra $\mathcal{A}$ with unit $\openone$ and a collection of unital $*$-subalgebras $\mathcal{A}(N) \subseteq \mathcal{A}$ indexed by regions $N\subseteq M$ with $\mathcal{A}(M) = \mathcal{A}$\footnote{Note that other authors (e.g.~\cite{Baer2007wave}) use $M$ together with \emph{precompact} regions as index set. Our choice, i.e., the choice of~\cite{fewster2018quantum}, is best-suited for the purpose of this paper, as it allows us to immediately consider the local algebras of non-precompact regions such as $M^\pm$, see Eq.~\eqref{eq_arrows}.}. For instance, as we will discuss in Sec.~\ref{Sec_local_particle_detectors}, the elements of the {polynomial field-$*$-algebra} of a region $N$ of the theory of the linear real scalar quantum field are algebraic combinations of smeared-out fields formally written as `$\int_N f(x) \phi(x) \; \mathrm{d}x$' for the quantum field $\phi$ and a test function $f$ {with compact support in} $N$. (Despite the fact that $\mathcal{A}(N)$ is sometimes called ``local algebra of observables of $N$'', obviously only the self-adjoint elements are considered to be observables.) The example of the smeared-ou{t} fields motivates the following model-independent axioms of a theory: 

\paragraph*{Isotony:} For regions $N_1 \subseteq N_2$: $\mathcal{A}(N_1) \subseteq \mathcal{A}(N_2)$.

\paragraph*{Einstein causality:} For spacelike separated regions $N_1$ and $N_2$: the elements of $\mathcal{A}(N_1)$ commute with the elements of $\mathcal{A}(N_2)$. 

\paragraph*{Time-slice property:} For regions $N_1 \subseteq N_2$, so that $N_1$ contains a Cauchy surface for $N_2$: $\mathcal{A}(N_1) = \mathcal{A}(N_2)$.\\

The time-slice property expresses the existence of an underlying local dynamical law. It is motivated by the idea that a quantum field should be determined by its data ``on'' (or rather \emph{around}) a Cauchy surface, hence any observable that is localisable in the domain of dependence (Cauchy development) of said Cauchy surface can be expressed in terms of observables ``on'' (or rather \emph{around}) it. Note that we are in the Heisenberg picture.

For the sake of completeness we mention that in the FV framework, every AQFT is also assumed to have a \emph{Haag property}{\footnote{Inspired by Haag duality~\cite{haag2012local} in the von Neumann setting: whenever an element $A \in \mathcal{A}$ commutes with all elements localisable in the causal complement of a compact set $K \subseteq M$, then $A \in \mathcal{A}(L)$ for every \emph{connected} region $L$ that contains $K$, see also footnote [13] in~\cite{bostelmann2020impossible}.\label{footnote_Haag}}}, which heuristically guarantees that the theory captures all relevant degrees of freedom, see Sec.~2 in~\cite{fewster2018quantum}. It is a mild additional assumption{\footnote{However, this as well as further desired properties might fail in free electromagnetism~\cite{DappiaggiLang2012,SandersDappiaggiHack2014}.}} that is for instance fulfilled by the linear real scalar quantum field as shown in Appendix~C in~\cite{fewster2018quantum}.

A \emph{state} on a theory (or simply on a unital $*$-algebra $\mathcal{A}$) is a linear map $\omega:\mathcal{A} \to \mathbb{C}$ which is normalised, i.e., $\omega(\openone)=1$ and positive, i.e., $\forall A \in\mathcal{A}:\; \omega(A^\dagger A) \geq 0$. We interpret it to assign expectation values to observables, see also~\cite{Drago_2020}. If the algebra elements are represented as operators on a Hilbert space, then a state $\omega$ could be of the from $\omega(A) = \Tr \qty(\rho_\omega A)$ for a ``density matrix'' $\rho_\omega$.

A state $\omega$ is called \emph{pure}, if for every two states $\omega_1, \omega_2$ and for every $\lambda \in (0,1)$ we have that $\omega = \lambda \omega_1 + (1- \lambda) \omega_2 \implies \omega_1 = \omega_2$. A state is called \emph{mixed} if it is not pure.

\subsubsection{Free combination of theories and entanglement}

Let us now discuss how we can combine two individually well-defined physical systems and how entanglement of those systems may be described. We take two unital $*$-algebras $\mathcal{A}_1$ and $\mathcal{A}_2$, for instance each the collection of local observables in region $N$ of two different theories on $M$. Their combination is given by the algebraic tensor product $\mathcal{A}_1 \otimes \mathcal{A}_2$, which is again a unital $*$-algebra. Additionally we may combine individual states to product states according to the following lemma.

\begin{lem}
Let $\sigma_j$ be states on $\mathcal{A}_j$, then functionals of the form $\sigma_1 \otimes \sigma_2$ are states on $\mathcal{A}_1 \otimes \mathcal{A}_2$ called \emph{product states}. 
\end{lem}

\begin{proof}
Normalisation is obvious. For positivity see~T.7 in Appendix T of~\cite{Wege-Olsen1993} or use Schur's product theorem.
\end{proof}

It is easy to see that product states show no correlation at all between the two systems $\mathcal{A}_1$ and $\mathcal{A}_2$. One may also consider finite statistical mixtures (finite convex combinations) of product states, i.e., $\sum_j \lambda_j \sigma_{1,j} \otimes \sigma_{2,j}$ for $\lambda_j >0$ with $\sum_j \lambda_j=1$, $j$ in some finite index set and even (appropriately normalised) pointwise limits thereof. This raises the question whether every state on $\mathcal{A}_1 \otimes \mathcal{A}_2$ arises in this way. In the case of {unital} $C^*$-algebras the answer is \emph{yes, if {and only if}} at least one of the two systems is classical, i.e., given by an Abelian {unital} $C^*$-algebra{\footnote{See Theorem~7 in~\cite{bacciagaluppi1993} (based on~\cite{raggio1988, baez1987}) and also Proposition 6 in~\cite{Baer2009} for a direct proof of one implication. See~\cite{landau1987} and the comment below Theorem 5.6 in~\cite{summers_independence} for a slight generalisation (in the case of von Neumann algebras).}.} This motivates the following definition.

\begin{mydef}
We call a state $\sigma$ on $\mathcal{A}_1 \otimes \mathcal{A}_2$ \emph{classically correlated}, if it can be written as the pointwise limit of convex combinations of states of the form $\sigma_1 \otimes \sigma_2$ for states $\sigma_j$ on $\mathcal{A}_j$. We call a state \emph{entangled}, if it is not classically correlated.
\end{mydef}

The above comment makes it clear that the presence of entangled states is phenomenon {of purely quantum composite systems}.

Let us remark that the natural combination of subsystems of a single theory, such as for instance $\mathcal{A}(N_1) \vee \mathcal{A}(N_2) \subseteq \mathcal{A}(M)$, the {unital} $*$-subalgebra generated by the two commuting algebras of local observables of two spacelike separated regions $N_1$ and $N_2$, might not always simply be isomorphic to a tensor product $\mathcal{A}(N_1) \otimes \mathcal{A}(N_2)$, see Sec.~VI in~\cite{summers_independence}. In order to not unnecessarily restrict the algebraic structure, we follow~\cite{VERCH_2005} and introduce the following notion.

\begin{mydef}
Given $\mathcal{A}, \mathcal{B} \subseteq \mathcal{R}$ two commuting, unital $*$-subalgebras of some unital $*$-algebra $\mathcal{R}$, we call a state $\omega$ on $\mathcal{A} \vee \mathcal{B} \subseteq \mathcal{R}$ \emph{product state}, if
\begin{equation}
    \forall A \in \mathcal{A}, \; \forall B \in \mathcal{B}: \; \omega(AB) = \omega(A) \omega(B).
\end{equation}
We call $\omega$ \emph{classically correlated}, if it is a pointwise limit of convex combinations of product states and \emph{entangled} if it is not classically correlated.
\label{def_generalised_bipartite_product_state}
\end{mydef}

We see that this encompasses the tensor product form above if $\mathcal{A}:= \mathcal{A}_1 \otimes {\mathbb{C}\openone}$, $\mathcal{B}:={\mathbb{C}\openone} \otimes \mathcal{A}_2$ and $\mathcal{R}= \mathcal{A}_1 \otimes \mathcal{A}_2$.

There seems to be no {simple} necessary and sufficient {criterion} that singles out entangled states in this very general setting, however, there are sufficient conditions, for instance the {(failure of the)} following \emph{qualitative} {one} for general {unital} $*$-algebras following the presentation in~\cite{VERCH_2005}.

\begin{mydef}[Verch-Werner ppt property, Definition~3.1 in~\cite{VERCH_2005}]
For $\mathcal{A}, \mathcal{B} \subseteq \mathcal{R}$ two commuting, unital $*$-subalgebras of some unital $*$-algebra $\mathcal{R}$, we say that a state $\omega$ on $\mathcal{A} \vee \mathcal{B} \subseteq \mathcal{R}$ has the \emph{Verch-Werner positive partial transpose (ppt)} property, if and only if for every $N \in \mathbb{N}$ and for all $x_1,..., x_N \in \mathcal{A}_1$ and $y_1,..., y_N \in \mathcal{A}_2$:
\begin{equation}
    \sum\limits_{j,k=1}^N \omega(x_k x_j^\dagger y_j^\dagger y_k) \geq 0.
    \label{eq_Verch-Werner_ppt}
\end{equation}
\end{mydef}
This is a generalisation of the positive partial transpose property of Peres~\cite{Peres_1996} and agrees with that notion at least for operators on finite-dimensional Hilbert spaces, see Proposition~3.2 in~\cite{VERCH_2005}.

Remark: For $N \in \{0, 1\}$ the condition is empty. 

\begin{lem}[after Lemma~3.3 in~\cite{VERCH_2005}]
Every classically correlated state $\sigma$ has the ppt property.
\end{lem}

It follows that every state that does not have the ppt property is entangled. The reversed implication does not hold in general. There are entangled states between finite-dimensional systems~\cite{Horodecki_1997, Horodecki1998bound}, as well as on the combination of two modes with two modes~\cite{Werner_2001} that \emph{have} the ppt property, hence a failing of the ppt property is only sufficient for entanglement. However, a state on the combination of one mode with one other mode is classically correlated \emph{if and only} if it has the ppt property~\cite{Simon_2000}, which we will utilise later.

\subsection{FV framework}

Let us now discuss how the coupling of local probes to a system of interest in a compact coupling zone is treated in the FV framework in a model-independent way.

The local probe is given in terms of a probe theory $\mathcal{P}$ and the system in terms of a system theory $\mathcal{S}$ on some globally hyperbolic spacetime $M$. The fully interacting combination of $\mathcal{S}$ and $\mathcal{P}$, i.e., the structure in which the probe is coupled to the system in a compact coupling zone $K$ shall itself be a theory{\footnote{The requirement that the interacting combination is itself a theory, in particular fulfilling Einstein causality, guarantees the locality of the coupling.}} and is denoted by $\mathcal{C}$. The fact that $\mathcal{C}$ is a coupled variant of the free combination $\mathcal{S} \otimes \mathcal{P}$ with interaction only switched on in $K$ is expressed by the existence of bijective, structure and localisation preserving identification maps $\mathcal{S}(N) \otimes \mathcal{P}(N) \to \mathcal{C}(N)$ \emph{outside} of the causal hull of $K$, i.e., for regions $N \subseteq M \setminus \ch(K)$, see Sec.~3.1 in~\cite{fewster2018quantum} for the details. For the covariantly defined \emph{in-region} $M^-:=M\setminus J^+(K)$ and \emph{out-region} $M^+=M \setminus J^-(K)$, this gives us the following maps:
\begin{equation}
    \begin{aligned}
    &\mathcal{S} \otimes \mathcal{P} \to \big(\mathcal{S} \otimes \mathcal{P}\big)(M^+) \to \mathcal{C}(M^+) \to \mathcal{C} \to \mathcal{C}(M^-) \to \big(\mathcal{S} \otimes \mathcal{P}\big)(M^-) \to \mathcal{S} \otimes \mathcal{P},
    \end{aligned}
    \label{eq_arrows}
\end{equation}
each of which is an isomorphism. The overall composition defines the \emph{scattering map} $\Theta: \mathcal{S} \otimes \mathcal{P} \to \mathcal{S} \otimes \mathcal{P}$, which is an automorphism. In the usual perturbative approach $\Theta$ is implemented as the adjoint action of the unitary scattering operator $\sf S$, i.e., $\Theta(A)={\sf S}^\dagger A{\sf S}$. An important property of $\Theta$ is given by the following lemma. 

\begin{lem}[Proposition 3.1 in~\cite{fewster2018quantum}] \hspace{1mm}
For every region $N \subseteq K^\perp: \Theta$ acts trivially on \linebreak $\big(\mathcal{S} \otimes \mathcal{P}\big)(N)$.
\label{lem_theta_spacelike_loc_change}
\end{lem}

Now suppose that the system is initially prepared in state $\omega$ and the probe in state $\sigma$. The effect of the interaction can be represented as an update of the tensor product of the initial states on the non-interacting combination according to 
\begin{equation}
    \omega \otimes \sigma \mapsto \Theta^*(\omega \otimes \sigma)=(\omega \otimes \sigma)\circ \Theta.
\end{equation}
The implementation of multiple probes in the FV framework was extensively discussed in~\cite{bostelmann2020impossible}. Based on that let us now assume that the probe theory $\mathcal{P}$ is the combination of two probes $\mathcal{P}_\mathsf{A}$ and $\mathcal{P}_\mathsf{B}$, so $\mathcal{P}_{\{\mathsf{A},\mathsf{B}\}} :=\mathcal{P}=\mathcal{P}_\mathsf{A} \otimes \mathcal{P}_\mathsf{B}$, each separately coupled to the system in coupling zones $K_\mathsf{A}$ and $K_\mathsf{B}$ with scattering maps $\Theta_\mathsf{A}: \mathcal{S} \otimes \mathcal{P}_\mathsf{A} \to \mathcal{S} \otimes \mathcal{P}_\mathsf{A}$ and $\Theta_\mathsf{B}: \mathcal{S} \otimes \mathcal{P}_\mathsf{B} \to \mathcal{S} \otimes \mathcal{P}_\mathsf{B}$ and initial state $\sigma_{\{\mathsf{A},\mathsf{B}\}} :=\sigma=\sigma_\mathsf{A} \otimes \sigma_\mathsf{B}$. We assume that \emph{causal factorisation}~\cite{fewster2018quantum, bostelmann2020impossible} holds, i.e.,
\begin{equation}
    \begin{aligned}
     \Theta = \begin{cases}
     \hat{\Theta}_\mathsf{A} \circ \hat{\Theta}_\mathsf{B} \qquad \text{for } K_\mathsf{A} \cap J^+(K_\mathsf{B}) = \emptyset,\\
     \hat{\Theta}_\mathsf{B} \circ \hat{\Theta}_\mathsf{A} \qquad \text{for } K_\mathsf{B} \cap J^+(K_\mathsf{A}) = \emptyset,
     \end{cases}
    \end{aligned}
    \label{eq_causal_fac}
\end{equation}
where $\hat{\Theta}_\mathsf{A} = \Theta_\mathsf{A} \otimes_3 \openone$ and $\hat{\Theta}_\mathsf{B} = \Theta_\mathsf{B} \otimes_2 \openone$. The subscript {of the tensor product} denotes the slot in which the second factor is inserted. We note that in the case of spacelike separated $K_\mathsf{A}, K_\mathsf{B}$, {it follows that} $\Theta= \hat{\Theta}_\mathsf{A} \circ \hat{\Theta}_\mathsf{B} = \hat{\Theta}_\mathsf{B} \circ \hat{\Theta}_\mathsf{A}$.

A feature of the FV framework is that any local probe observable $C \in \mathcal{P}$ induces a system observable $\varepsilon(C) \in \mathcal{S}$ such that
\begin{equation}
    \omega(\varepsilon(C)) = (\omega \otimes \sigma)(\Theta(\openone \otimes C)).
    \label{eq_induced_obs}
\end{equation}
One can even make a statement about the localisation of the induced observable, i.e., for every connected region $L$ that contains the coupling zone $K$, it follows that $\varepsilon \qty[\mathcal{P}(M)] \subseteq \mathcal{S}(L)$. In other words, every induced observable can be localised in any connected region around the coupling zone, see Theorem 3.3 in~\cite{fewster2018quantum} for details{\footnote{The connectedness condition on $L$ (here and below) stems from the formulation of the Haag property, see footnote~\ref{footnote_Haag}.}}.

Moreover, in~\cite{bostelmann2020impossible} (see Eq.~(42) therein) it was shown that in the case of a bipartite probe $\mathcal{P}_{\{ \mathsf{A},\mathsf{B}\}}$ with spacelike separated coupling zones $K_\mathsf{A}, K_\mathsf{B}$ we have for every system state $\omega$
\begin{equation}
    \begin{aligned}
     \forall A \in \mathcal{P}_\mathsf{A}\;  \forall B \in \mathcal{P}_\mathsf{B}: \omega \qty(\varepsilon_{\{\mathsf{A}, \mathsf{B}\}}(A \otimes B)) &=  \omega\qty(\varepsilon_{\mathsf{A}}(A) \varepsilon_{\mathsf{B}}(B))= \omega\qty(\varepsilon_{\mathsf{B}}(B) \varepsilon_{\mathsf{A}}(A)).
    \end{aligned}
    \label{eq_induced_obs_spacelike}
\end{equation}

\section{Model-independent entanglement harvesting with local probes}
\label{Sec_Model_indep_E_H}

Let us consider a theory of interest, the \emph{system} $\mathcal{S}$, and two agents, Alice ($\mathsf{A}$) and Bob ($\mathsf{B}$), {with associated} probe theories $\mathcal{P}_\mathsf{A}$ and $\mathcal{P}_\mathsf{B}$ respectively. The agents couple their probes to the system in connected, compact coupling zones $K_\mathsf{A}$ and $K_\mathsf{B}$ respectively, which are assumed to be causally orderable, i.e., $ K_\mathsf{A} \cap J^+(K_\mathsf{B}) = \emptyset$ or $ K_\mathsf{B} \cap J^+(K_\mathsf{A}) = \emptyset$. The coupling gives rise to a scattering map $\Theta_\mathsf{A}: \mathcal{S} \otimes \mathcal{P}_\mathsf{A} \to \mathcal{S} \otimes \mathcal{P}_\mathsf{A}$ for $\mathsf{A}$, similar for $\mathsf{B}$. As before, we assume that causal factorisation holds, i.e., that there is a theory that describes the coupling of both of the two probes to the system with overall scattering map $\Theta$ that fulfills Eq.~\eqref{eq_causal_fac}. 

Let us suppose that the initial state of the combination of the two probes $\mathcal{P}_\mathsf{A} \otimes \mathcal{P}_\mathsf{B}$ is uncorrelated before the interaction with the system and hence given by $\sigma:=\sigma_\mathsf{A} \otimes \sigma_\mathsf{B}$ for $\sigma_\mathsf{I}$ a state on $\mathcal{P}_\mathsf{I}$. Let $\omega$ be the initial state on $\mathcal{S}$. The effect of the interaction of the probes with the system can be represented on the non-interacting combination as an update of the initial product state $\sigma$ to
\begin{equation}
    \begin{aligned}
     \sigma'(C):=(\omega \otimes \sigma)(\Theta(1\!\!\!\!1 \otimes C)).
    \end{aligned}
    \label{eq_updated_state}
\end{equation}
To emphasise that the expression in Eq.~\eqref{eq_updated_state} is exactly what one would expect, let us assume that the algebra elements act on a Hilbert space, that $\omega,  \sigma$ are ``density matrices'' $\rho_\omega,  \rho_\sigma$ and $\Theta(A) = \mathsf{S}^\dagger A \mathsf{S}$. Then $\sigma'(C) = \Tr_{\mathcal{S} + \mathcal{P}}\qty(\rho_\omega \otimes \rho_\sigma \qty(\mathsf{S}^\dagger \openone \otimes C \mathsf{S})) = \Tr_\mathcal{P} \qty(\rho_{\sigma}' C)$ where $\rho_{\sigma}'= \Tr_\mathcal{S}\qty( \mathsf{S} \qty(\rho_\omega \otimes \rho_\sigma)\mathsf{S}^\dagger)$.

Recalling the definition of induced observables in Eq.~\eqref{eq_induced_obs}, it follows immediately that
\begin{equation}
    \begin{aligned}
    \sigma'(C) = \omega(\varepsilon_{\{\mathsf{A},\mathsf{B}\}}(C)).
    \label{eq_updated_state_varepsilon}
    \end{aligned}
\end{equation}
Entanglement harvesting describes the process of ``harvesting'' entanglement from the system state $\omega$ and transferring it to the state $\sigma'$ on the combination of the two probes. The two agents Alice and Bob may access and process the information in the updated state $\sigma'$ in local regions of control $N_\mathsf{A}$ and $N_\mathsf{B}$ respectively. We want to think of those regions as ``processing regions'', e.g., Alice couples her probe theory to the system in the compact coupling zone $K_\mathsf{A}$ and then analyses the updated state on her probe in some spacetime region $N_\mathsf{A}$. It is clear that we want to choose $N_\mathsf{A}$ outside the causal past of $K_\mathsf{A}$, so $N_\mathsf{A} \subseteq M^+_\mathsf{A} := M \setminus J^-(K_\mathsf{A})$, similar for Bob, where $N_\mathsf{B} \subseteq M^+_\mathsf{B} := M \setminus J^-(K_\mathsf{B})$. The ``analysis'' that we have in mind is aimed at detecting possible correlation between the local algebras $\mathcal{P}_\mathsf{A}(N_\mathsf{A})$ and $\mathcal{P}_\mathsf{B}(N_\mathsf{B})$ respectively in the state $\sigma'$. The initial state $\sigma= \sigma_\mathsf{A} \otimes \sigma_\mathsf{B}$ is a product state on $\mathcal{P}_\mathsf{A} \otimes \mathcal{P}_\mathsf{B}$ and hence, after restriction, also a product state on $\mathcal{P}_\mathsf{A}(N_\mathsf{A}) \otimes \mathcal{P}_\mathsf{B}(N_\mathsf{B})$. This suggests that any correlation between these two parties in the state $\sigma'$ must have come from the interaction with the system, i.e. must have been ``harvested'' from the system.

Before we continue, let us repeat what we have already mentioned in the introduction: the entanglement harvesting protocol with non-local quantum mechanical particle detectors follows a very similar approach. The main difference is that due to the simple internal dynamics of a non-local particle detector, the analysis that takes place after the interaction (i.e., in a processing region $N_\mathsf{A}$) can comprise all of the probe's degrees of freedom, whereas for a truly local probe, agent $\mathsf{A}$ may only access those degrees of freedom that can be localised in $N_\mathsf{A}$.

It is not surprising that the correlations that Alice and Bob observe depend on the location of the processing regions $N_\mathsf{A}, N_\mathsf{B}$, as shown in the following theorem, whose proof can be found in \ref{sec_appendix_proof_theo_regions_of_sep}.

\begin{theo}
Without loss of generality assume that $K_\mathsf{A} \cap J^+(K_\mathsf{B}) = \emptyset$, then $\sigma'$ is a product state on $\mathcal{P}_\mathsf{A}(N_\mathsf{A}) \otimes \mathcal{P}_\mathsf{B}(N_\mathsf{B})$, for regions $N_\mathsf{A}, N_\mathsf{B}$ whenever either of the following holds
    \begin{enumerate}
        \item $\ch(N_\mathsf{A} \cup N_\mathsf{B}) \subseteq (K_\mathsf{A} \cup K_\mathsf{B})^\perp$,
        \item\label{caseII_theo_regions_of_separability} or $N_\mathsf{B} \subseteq K_\mathsf{B}^\perp$ and $N_\mathsf{A}$ is arbitrary,
        \item\label{caseIII_theo_regions_of_separability} or $N_\mathsf{A} \subseteq K_\mathsf{A}^\perp$, and $N_\mathsf{B} \subseteq K_\mathsf{A}^\perp \cap M_\mathsf{B}^+$ and $N_\mathsf{B}$ is precompact\footnote{A set $N \subseteq M$ is called \emph{precompact}, if its closure $\overline{N}$ is compact in $M$. In particular, any subset of a compact set is precompact.}. 
    \end{enumerate}
\label{theo_regions_of_separability}
\end{theo}

We emphasise that~\eqref{caseIII_theo_regions_of_separability} is merely a special {form} of~\eqref{caseII_theo_regions_of_separability} \emph{only} if the coupling zones are spacelike separated (i.e., {if} \emph{both} $K_\mathsf{A} \cap J^+(K_\mathsf{B}) = \emptyset$ and $K_\mathsf{B} \cap J^+(K_\mathsf{A}) = \emptyset$). {In this case,} we \emph{must} take \emph{each} $N_\mathsf{I}$ to lie in $J^+(K_\mathsf{I})$ (or at least have non-trivial intersection with it) in order to harvest any correlation. {This is in agreement with causality~\cite{Buchholz1994,bostelmann2020impossible}.}

{If $K_\mathsf{A}$ and $K_\mathsf{B}$ are spacelike separated,} it is {also} intuitively plausible that whether the updated state $\sigma'$ can exhibit any correlation between regions $N_\mathsf{A}, N_\mathsf{B}$ depends on the initial system state $\omega$. Indeed, for an arbitrary observable $C \in \mathcal{P}_\mathsf{A} \otimes \mathcal{P}_\mathsf{B}$ and spacelike separated $K_\mathsf{A}, K_\mathsf{B}$ we see, using Eq{s}.{~\eqref{eq_updated_state_varepsilon} and}~\eqref{eq_induced_obs_spacelike}, that
\begin{equation}
    \begin{aligned}
    \sigma'(C) &=\sum_j \omega(\varepsilon_{\{\mathsf{A} ,\mathsf{B}\}}(A_j \otimes B_j))=\sum_j \omega(\varepsilon_{\mathsf{A} }(A_j)\varepsilon_{\mathsf{B}}(B_j)),
    \end{aligned}
\end{equation}
where we wrote $C=\sum_j A_j \otimes B_j$, for $j$ running through a finite index set. Let $L_\mathsf{A}$ be a connected region containing $K_\mathsf{A}$ and $L_\mathsf{B}$ be a connected region containing $K_\mathsf{B}$ such that $L_\mathsf{A}$ and $L_\mathsf{B}$ are spacelike separated. Let us assume that $\omega$ restricted to $\mathcal{S}(L_\mathsf{A}) \vee \mathcal{S}(L_\mathsf{B}) \subseteq \mathcal{S}$ is a product state according to Definition~\ref{def_generalised_bipartite_product_state}. Then, since $\varepsilon_{\mathsf{A} }(A_j) \in \mathcal{S}(L_\mathsf{A})$ and $\varepsilon_{\mathsf{B} }(B_j) \in \mathcal{S}(L_\mathsf{B})$, it follows that 
\begin{equation}
    \begin{aligned}
    \sigma'(C) &=\sum_j \omega(\varepsilon_{\mathsf{A} }(A_j)) \; \omega(\varepsilon_{\mathsf{B} }(B_j))= \sum_j \sigma_\mathsf{A}'(A_j) \sigma_\mathsf{B}'(B_j) = (\sigma_\mathsf{A}'\otimes \sigma_\mathsf{B}')(C),
    \end{aligned}
\end{equation}
where we used that for $\sf J \in \{A,B\}$ $\sigma_\mathsf{J}'(Z) := \omega(\varepsilon_{\mathsf{J} }(Z))$ is obviously a state. We can proceed similarly if $\omega$ is a convex combination of product states and hence proved the following theorem.

\begin{theo}
Assume that $K_\mathsf{A} \subseteq L_\mathsf{A}$ and $K_\mathsf{B} \subseteq L_\mathsf{B}$ for spacelike separated connected regions $L_\mathsf{A}, L_\mathsf{B}$. Then the following holds: If $\omega$ is a product state, then $\sigma'$ is a product state as well. If $\omega$ is classically correlated on $\mathcal{S}(L_\mathsf{A}) \vee \mathcal{S}(L_\mathsf{B})$, then $\sigma'$ is classically correlated on $\mathcal{P}_\mathsf{A} \otimes \mathcal{P}_\mathsf{B}$.
\end{theo}

In the investigation above we explicitly saw that the separability of $\sigma'$ depends on the separability of $\omega$ on the {unital} $*$-algebra spanned by the union of the images of $\varepsilon_\mathsf{J}$ \emph{only}. In general, this might be difficult to handle because the images of $\varepsilon_\mathsf{J}$ might not be {unital} $*$-algebras. However, after adding this additional assumption (which holds in the explicit model constructed in~\cite{fewster2018quantum}, see the end of Sec.~5.1 therein), we get the slightly stronger following statement.

\begin{theo}
{Assume that $K_\mathsf{A} \subseteq L_\mathsf{A}$ and $K_\mathsf{B} \subseteq L_\mathsf{B}$ for spacelike separated connected regions $L_\mathsf{A}, L_\mathsf{B}$. Additionally}, let us assume that (at least for fixed processing regions $N_\mathsf{A}, N_\mathsf{B}$ and) for $\sf J \in \{ A,B\}$ we have that $\varepsilon_\mathsf{J}\qty[\mathcal{P}_\mathsf{J}(N_\mathsf{J})]$ is a {unital} $*$-algebra. Then:
\begin{enumerate}
    \item If $\omega$ is a product state [classically correlated] on $\varepsilon_\mathsf{A}\qty[\mathcal{P}_\mathsf{A}(N_\mathsf{A})] \vee \varepsilon_\mathsf{B}\qty[\mathcal{P}_\mathsf{B}(N_\mathsf{B})]$, then $\sigma'$ is a product state [classically correlated] on $\mathcal{P}_\mathsf{A}(N_\mathsf{A})\otimes \mathcal{P}_\mathsf{B}(N_\mathsf{B})$ as well.
    \item If $\sigma'$ is a product state on $\mathcal{P}_\mathsf{A}(N_\mathsf{A})\otimes \mathcal{P}_\mathsf{B}(N_\mathsf{B})$, then $\omega$ is a product state on $\varepsilon_\mathsf{A}\qty[\mathcal{P}_\mathsf{A}(N_\mathsf{A})] \vee \varepsilon_\mathsf{B}\qty[\mathcal{P}_\mathsf{B}(N_\mathsf{B})]$.
\end{enumerate}
\end{theo}

\begin{proof}
The first statement immediately follows from the earlier discussion. For the second statement assume that $\sigma'=\sigma_\mathsf{A}' \otimes \sigma_\mathsf{B}'$ and take $\tilde{A} \in \varepsilon_\mathsf{A}\qty[\mathcal{P}_\mathsf{A}(N_\mathsf{A})]$ and  $\tilde{B} \in \varepsilon_\mathsf{B}\qty[\mathcal{P}_\mathsf{A}(N_\mathsf{A})]$, then
\begin{equation}
    \begin{aligned}
    \omega(\tilde{A} \tilde{B}) &= \omega(\varepsilon_{\mathsf{A}}(A) \varepsilon_{\mathsf{B}}(B)) = \sigma'(A \otimes B)= \sigma'(A \otimes \openone) \sigma'(\openone \otimes B)\\
    &= \omega(\varepsilon_{\mathsf{A}}(A))\omega(\varepsilon_{\mathsf{B}}(B))=\omega(\tilde{A}) \omega(\tilde{B}),
    \end{aligned}
\end{equation}
where we used Eq.~\eqref{eq_induced_obs_spacelike}.
\end{proof}
These two results show that any potentially harvested correlation of the state $\sigma'$ must come from the correlation of the spacelike separated algebras localisable in regions around the coupling zones in the initial system state $\omega$. However, they just give a \emph{necessary} condition for entanglement harvesting and do not guarantee that $\sigma'$ is entangled when $\omega$ is. A general \emph{sufficient} condition ({for example on} the coupling) is surely desirable, however, we are not aware of such a statement.

\section{Local particle detectors}
\label{Sec_local_particle_detectors}

After the general summary of model-independent local probes in Sec.~\ref{Sec_local_probes} let us now discuss a class of explicit probe models introduced in~\cite{fewster2018quantum}. Here, the local probe is given by a linear real scalar field (possibly under influence of external fields). We show how it is possible to restrict our attention to a single local mode of the scalar field, which basically forms a reduction from infinitely many to just one single degree of freedom. We argue that such a local mode is a realistic model for a particle detector. 

We start by recalling the algebraic, covariant quantisation of the linear real scalar field on {a} (possibly curved) globally hyperbolic spacetime and then discuss the restriction of the field to a local mode.

\subsection{Linear real scalar quantum field}

A \emph{classical} linear real scalar field on a globally hyperbolic spacetime $M$ is defined by a normally hyperbolic equation of motion (eom) $P \varphi= 0$ (for instance $P= \Box + m^2$ for the Klein-Gordon field\footnote{Our signature convention is mostly minuses, i.e., $(+,-,\dots,-)$.})~\cite{Baer2007wave}. There are similar yet \emph{in}-equivalent ways of assigning an algebraic quantum field theory (according to Sec.~\ref{Sec_local_probes}) to this equation of motion. We define the {polynomial field-$*$-algebra} $\mathcal{F}$, which can be written down by products and sums of the identity $\openone$ and formal symbols (``smeared fields'') $\varphi(f)$, for $f \in C_c^\infty(M;\mathbb{C})$ a smooth, compactly supported complex-valued function. (The reader might find it helpful to \emph{think} of $\varphi(f)$ as $\int \varphi(x) f(x) \mathrm{d}x$.) They fulfill the following properties\footnote{See for instance Appendix B in~\cite{fewster2019algebraic} for a recipe how the algebra $\mathcal{F}$ is constructed from its generators and relations.}
\begin{enumerate}
    \item $f \mapsto \varphi(f)$ is $\mathbb{C}$-linear,
    \item $\varphi(f)^\dagger = \varphi(\overline{f})$,
    \item $\varphi(Pf) =0$,
    \item $\qty[\varphi(f_1), \varphi(f_2)]= \mathrm{i} E(f_1,f_2) \openone$,
\end{enumerate}
where $\overline{f}$ denotes the complex conjugation of $f$, $E(f_1,f_2):= \int_M f_1 (Ef_2) \mathrm{d}V_M$ for the causal propagator $E$ (the difference of the advanced and retarded Green operators\footnote{It is common to write the action of $E$ on a test function $f \in C_c^\infty(M;\mathbb{C})$ in terms of a (possibly distributional) ``integral kernel'' $E(x,y)$, i.e., $(Ef)(x) = \int E(x,y) f(y) \; \mathrm{d}V_M(y)$.}) associated to the normally hyperbolic equations of motion~\cite{Baer2007wave}. The local algebra of a region $N$, $\mathcal{F}(N)$, is then defined to be the unital $*$-algebra generated by $\{\varphi(f) | f \in C_c^\infty(N; \mathbb{C})\}$ and $\openone$, i.e., $\mathbb{C}$-linear combinations and finite products and adjoints.

It was shown in~\cite{fewster2018quantum} that $\mathcal{F}$ fulfils all the axioms of an AQFT, in particular the Haag property (see Appendix C therein) as well as the time-slice property, which immediately follows from the properties of the underlying normally hyperbolic equations of motion.

We emphasise, that $\mathcal{F}$ contains only finite linear combinations and products of smeared fields, so for instance expressions such as $\varphi(f) \varphi(f)$ (naively $\int \varphi(x) \varphi(y) f(x) f(y)\mathrm{d}x \mathrm{d}y$), but no Wick-ordered expressions such as $:\varphi^2(f):$ (naively $\int :\varphi^2(x): f(x) \mathrm{d}x$). However, the {polynomial field-$*$-algebra} is the starting point of defining Wick-ordered expressions and of constructing interacting quantum field theories perturbatively. We also emphasise that this construction is a fully \emph{covariant} quantisation of the linear real scalar field; no Cauchy surface was chosen, no $s+1$-split of the manifold $M$ was made\footnote{The relationship between (spacetime) smeared fields $\varphi$ and (space) smeared equal-time canonically conjugate variables $\hat{\Phi}, \hat{\Pi}$ on some Cauchy surface $\Sigma$ is given by the following expression (see~Theorem 2 in~\cite{Dimock1980} and also Eq.~(10) in~\cite{HollandsWald2015quantum}): $\varphi(f) = \int_\Sigma \big( (\nabla_n E f\restriction_\Sigma (\vec{x})) \;  \hat{\Phi}(\vec{x}) -  (E f\restriction_\Sigma(\vec{x})) \;  \hat{\Pi}(\vec{x})\big)\;  \mathrm{d}V_\Sigma(\vec{x})$.}. 

Another way of defining an AQFT for the linear real scalar field is in terms of the CCR-$C^*$- algebra spanned by Weyl-generators $W(f)$ for $f \in C_c^\infty(M;\mathbb{R})$. (The reader might find it helpful to \emph{think} of $W(f)$ as exponentiated smeared fields $W(f) = e^{\mathrm{i} \varphi(f)}$.) They fulfill
\begin{enumerate}
    \item $W(f)^\dagger = W(-f)$,
    \item $W(Pf) = \openone$,
    \item $W(f) W(g) = e^{-\frac{\mathrm{i}}{2} E(f,g)} W(f+g)$,
\end{enumerate}
so in particular $W(0) = \openone$ and $W(f)^{\dagger}= W(f)^{-1}$. The abstract {unital} $*$-algebra spanned by the $W(f)$'s can be equipped with a unique $C^*$-norm and the resulting {unital} $C^*$-algebra is called the (\emph{canonical commutation relations}) CCR-$C^*$-algebra, see Sec.~1.6 in~\cite{Baer2009}. It gives rise to an AQFT $\mathcal{A}$, where $\mathcal{A}(N)$ is the $C^*$-closure of all finite linear combinations of $W(f)$ for $f \in C_c^\infty(N;\mathbb{R})$\footnote{The validity of the Haag property is not obvious and will be discussed in future work.}.

The advantage of introducing the CCR-$C^*$-algebra is that one can now apply the well-developed $C^*$-representation theory, as we will do in Sec.~\ref{sec_subsec_Reeh-Schlieder_mixed}. Moreover, it forms the starting point for introducing von Neumann algebras, which allow for an even more powerful structural analysis.

The field $*$-theory $\mathcal{F}$ and the CCR-$C^*$-theory $\mathcal{A}$ really are different. The former does not contain Weyl generators, the latter does not contain (smeared) fields. Despite the difference between the two theories, there is a distinguished class of states on each of them called \emph{quasi-free} or also \emph{Gaussian} states \emph{with possibly non-vanishing one-point function} that are in a natural correspondence. However, this is not necessarily true for all states. 

Quasi-free states on $\mathcal{A}$ with possibly non-vanishing one-point function have the form
\begin{equation}
    \begin{aligned}
    \hat{\omega}(W(f))=e^{\mathrm{i} \chi(f)- \frac{1}{4} \beta(f,f)},
    \end{aligned}
\end{equation}
for $f \in C_c^\infty(M;\mathbb{R})$, $\chi: C_c^\infty(M;\mathbb{R) \to \mathbb{R}}$ real-linear and $\beta: C_c^\infty(M; \mathbb{R}) \times C_c^\infty(M; \mathbb{R}) \to \mathbb{R}$ a symmetric,  $\mathbb{R}$-bilinear form that fulfills the following positivity condition (independent of $\chi$)
\begin{equation}
    \begin{aligned}
    \forall f,g \in C_c^\infty(M; \mathbb{R}): |E(f,g)|^2 \leq \beta(f,f) \beta(g,g),
    \end{aligned}
    \label{eq_positivity condition}
\end{equation}
see~\cite{fewster2019algebraic, petz1990invitation} for details\footnote{In particular we also need to demand that $\beta$ gives rise to a positive semidefinite symmetric bilinear form $\tilde{\beta}: C_c^\infty(M;\mathbb{R})/P C_c^\infty(M;\mathbb{R}) \times C_c^\infty(M;\mathbb{R})/PC_c^\infty(M;\mathbb{R}) \to \mathbb{R}$ via $\tilde{\beta}(Ef,Eg) = \beta(f,g)$ for Eq.~\eqref{eq_two-point_func} to be consistent. Similarly for $\tilde{\chi}(Ef) = \chi(f)$.}. The corresponding quasi-free state $\omega$ on the {unital} $*$-algebra $\mathcal{F}$ can be defined by identities of the form
\begin{equation}
    \begin{aligned}
    \sum\limits_{n=0}^\infty x^n \frac{\mathrm{i}^n}{n!} \omega(\varphi(f)^n) = e^{\mathrm{i} x \chi(f)- \frac{1}{4} x^2 \beta(f,f)}
    \end{aligned}
    \label{eq_n-pt-f_from_exp}
\end{equation}
between formal $\mathbb{C}$-valued power-series in the formal parameter $x$. It is then easy to see that $\forall f,g \in C_c^\infty(M; \mathbb{R}):$
\begin{equation}
    \begin{aligned}
    \omega(\varphi(f)) &= \chi(f),\\
    \omega(\varphi(f) \varphi(g)) &= \frac{1}{2}\beta(f,g) + \chi(f) \chi(g) + \frac{1}{2} \mathrm{i} E(f,g),\\
    \omega\qty(\qty{\varphi(f),\varphi(g)}) &= \beta(f,g) + 2 \chi(f) \chi(g),
    \end{aligned}
    \label{eq_two-point_func}
\end{equation}
where $\qty{\varphi(f),\varphi(g)} = \varphi(f)\varphi(g) + \varphi(g)\varphi(f)$ is the anti-commutator. One can also define the so-called truncated two-point function via
\begin{equation}
    \begin{aligned}
    \omega(\varphi(f) \varphi(g)) -\chi(f) \chi(g)  &= \frac{1}{2}\beta(f,g) + \frac{1}{2} \mathrm{i} E(f,g),
    \end{aligned}
\end{equation}
so $\frac{1}{2} \beta$ is the symmetric part of the truncated two-point function. It is the nature of a quasi-free state that all truncated $n$-point functions for $n >2$ vanish, which means that all higher $n$-point functions can be calculated from the one- and two-point function.

\subsection{Local modes}

What is usually meant by the ``infinite degrees of freedom'' of a quantum field is that the symplectic space $\qty(C_c^\infty(M;\mathbb{R})/P C_c^\infty(M;\mathbb{R}),E(\cdot,\cdot))$ (that is underlying a CCR-$C^*$-quantisation for instance according to~\cite{Baer2007wave}) is infinite-dimensional. In order to reduce the complexity and model a more realistic situation, in which only finitely many degrees of freedom are accessible by one observer, we restrict our attention to only one mode of the field that can be localised in a \emph{finite}\footnote{{In fact we will consider/allow all regions $N$, whose causal complement contains a region, see Lemma~\ref{lem_restriction_mixed}. If the Cauchy surfaces of $M$ are not compact, then this comprises all precompact regions. (For precompact $N$, $\overline{N}^\perp \subseteq N^\perp$ is either empty or a region. If it is empty, then $M= J^+(\overline{N}) \cup J^-(\overline{N})$, hence is spacially compact and hence every Cauchy surface $\Sigma=\qty(J^+(\overline{N}) \cup J^-(\overline{N})) \cap \Sigma$ is compact, see Lemma~1.5 in~\cite{baer2015green}.)}} region. Our motivation is that a realistic observer should only have access to such \emph{local} degrees of freedom. The standard annihilation and creation operators (for sharp momentum in Minkowski spacetime for instance) cannot be used for this, because the associated mode is not localisable in a finite region of spacetime and hence physically speaking not accessible. Instead we use a \emph{local} mode, whose construction shall be explained now.

Let us look at the quantum field $\varphi$ and let us pick $f_1, f_2 \in C_c^\infty(N;\mathbb{R})$, two \emph{real}-valued smooth functions with compact support contained in a region $N$ such that {$N^\perp$ contains a region and} $E(f_1,f_2) \neq 0$. Then without loss of generality we can assume that $E(f_1,f_2)=1$. Using suggestive notation, we define $Q:= \varphi(f_1)$ and $P:= \varphi(f_2)$ and hence have 
\begin{equation}
    \begin{aligned}
    \qty[Q,P]=\qty[\varphi(f_1), \varphi(f_2)]=\mathrm{i} \openone.
    \end{aligned}
\end{equation}
The unital $*$-algebra spanned by $\varphi(f_1), \varphi(f_2)$ and $\openone$ then describes the \emph{local} mode of the probe field $\varphi$ localisable in $N$ defined by $f_1,f_2 \in C_c^\infty(N;\mathbb{R})$. The term ``local mode'' will always refer to such a subalgebra of the {polynomial} field-$*$-algebra. Likewise, the CCR-$C^*$-algebra of this local mode is the {unital} $C^*$-algebra spanned by $W(f_1)$ and $W(f_2)$.

Similar to how the mode of a non-relativistic harmonic oscillator is used as a \emph{non-local} particle detector, we can use this local mode of the scalar field as a \emph{local} particle detector. To that end let us set $f:= \frac{1}{\sqrt{2}}(f_1 + \mathrm{i} f_2) \in C_c^\infty(N;\mathbb{C})$ and
\begin{equation}
    \begin{aligned}
    a&:= \frac{1}{\sqrt{2}}\qty(Q + \mathrm{i}P) = \varphi\qty(f), \qquad a^\dagger&:= \frac{1}{\sqrt{2}}\qty(Q - \mathrm{i}P) = \varphi\qty(\overline{f}),
    \end{aligned}
\end{equation}
which are the (local) annihilation and creation operators associated to {the parametrisation of} this single local mode under consideration {in terms of $f_1, f_2$}.

The restriction of a quasi-free state $\omega$ of the whole scalar field to the single mode of interest is a quasi-free state as well with so-called $2\times2$ \emph{covariance} matrix 
\begin{equation}
    \begin{aligned}
    A_{jk}:=\omega(\qty{\varphi(f_j), \varphi(f_k)}) - 2 \omega(\varphi(f_j)) \omega(\varphi(f_k)) = \beta(f_j,f_k).
    \end{aligned}
\end{equation}
The positivity condition in Eq.~\eqref{eq_positivity condition} restricted to the one mode takes the from
\begin{equation}
    \begin{aligned}
    A + \mathrm{i} s \geq 0,
    \end{aligned}
    \label{eq_positivity_condition_mode}
\end{equation}
where $s_{jk}:=E(f_j,f_k)=\mqty(0 & 1 \\ -1& 0)_{jk}$. This is also known as the \emph{uncertainty principle}. We have in particular that $\det(A) \geq 1$.

Before we continue let us recall the following general result, whose proof is in~\ref{sec_appendix_pure_quasi-free_det}.

\begin{lem}
Let $\omega$ be a quasi-free state on one mode with covariance matrix $A$. Then $\det(A)=1$ if and only if $\hat{\omega}$ is a pure state on the corresponding CCR-$C^*$-algebra.
\label{lem_pure_det}
\end{lem}

Associated to the local annihilation and creation operators is the number operator $a^\dagger a$ and for the quasi-free state $\omega$ we have
\begin{equation}
    \begin{aligned}
    \omega(a^\dagger a) = \frac{1}{2} \qty(\underbrace{ \frac{1}{2} \Tr \qty(A) - 1}_{\geq 0} + \underbrace{\omega(\varphi(f_1))^2 + \omega(\varphi(f_2))^2}_{\geq 0}).
    \end{aligned}
\end{equation}
$\omega$ restricts to the ground state of the local {number operator} if and only if both non-negative summands vanish. Together with the fact that $1 \leq \det(A)$ and that $A$ is real symmetric, we see that the first summand vanishes if and only if $A = \openone$. In particular it follows from the previous lemma that $\omega$ cannot restrict to the ground state {of any number operator\footnote{Note that the operators $a, a^\dagger, a^\dagger a$ depend on the parametrisation of a mode and are not unique.} of a local mode} if $\hat{\omega}$ is truly mixed.

\subsection{Reeh-Schlieder states restricted to single local modes}
\label{sec_subsec_Reeh-Schlieder_mixed}

In the present section we investigate whether physically reasonable states $\omega$ of the quantum field $\varphi$ restrict to the ground state, or more generally, whether the associated states $\hat{\omega}$ restrict to pure states on a local mode. To that end we will first discuss some abstract technical results and then apply those to the CCR-$C^*$-algebra $\mathcal{A}$ and the CCR-$C^*$-algebra of a local mode.

As mentioned earlier, many physically reasonable states show entanglement over spacelike separation. This entanglement is often a consequence of the \emph{Reeh-Schlieder} property of those states.

\begin{mydef}[Reeh-Schlieder property I]
Let $\mathcal{A}$ be a {unital} $C^*$-algebra of operators on a Hilbert space $\mathcal{H}$ and let $\psi \in \mathcal{H}$ be a unit vector. Denote by $\omega_\psi$ the state $\omega_\psi(\cdot) := \braket{\psi| \cdot \psi}$, then we say that $\omega_\psi$ has the Reeh-Schlieder property with respect to $\mathcal{A}$ if $\psi$ is a cyclic vector for $\mathcal{A}$, i.e.,
\begin{equation}
\forall \xi \in \mathcal{H} \; \exists \qty(A_n)_{n\in \mathbb{N}} \subseteq \mathcal{A}: \lim\limits_{n \to \infty} A_n \psi = \xi.
\end{equation}
\end{mydef}

For an abstract {unital} $C^*$-subalgebra $\mathcal{A} \subseteq \mathcal{R}$ of a {unital} $C^*$-algebra $\mathcal{R}$ with state $\omega$ on $\mathcal{R}$ we may look at the GNS representation of {$\mathcal{R}$} (see Definition~13 in~\cite{Baer2009}) $\pi_\omega: \mathcal{R} \to BL(\mathcal{H}_\omega)$, which (is possibly not injective and) maps $\mathcal{A}$ to a {unital} $C^*$-subalgebra $\pi_\omega\qty[\mathcal{A}]$ of the bounded linear operators $BL(\mathcal{H}_\omega)$ on a ($\omega$-dependent) complex Hilbert space $(\mathcal{H}_\omega, \braket{\cdot|\cdot})$. Moreover there exists a unit vector $\Omega_\omega \in \mathcal{H}_\omega$ such that $\forall A \in \mathcal{A}$ we have that $\omega(A) = \braket{\Omega_\omega | \pi_\omega(A) \Omega_\omega}$.

\begin{mydef}[Reeh-Schlieder property II]
For an abstract {unital} $C^*$-subalgebra $\mathcal{A} \subseteq \mathcal{R}$ of a {unital} $C^*$-algebra $\mathcal{R}$ with state $\omega$ on $\mathcal{R}$ we say that $\omega$ has the Reeh-Schlieder property with respect to $\mathcal{A}$ if $\Omega_\omega \in \mathcal{H}_\omega$ is a cyclic vector for $\pi_\omega\qty[\mathcal{A}]$.
\end{mydef}

The following result about entanglement as a consequence of the Reeh-Schlieder property is taken from~\cite{VERCH_2005}.

\begin{lem}[Theorem~6.2 in~\cite{VERCH_2005}]
Let $\mathcal{A}, \mathcal{B} \subseteq \mathcal{R}$ be two commuting, non-Abelian{, unital} $C^*$-subalgebras of some {unital} $C^*$-algebra $\mathcal{R}$ realised as operators on some Hilbert space $\mathcal{H}$ and let $\psi \in \mathcal{H}$ be a unit vector. Then 
\begin{equation}
\begin{aligned}
    &\omega_\psi \text{ has the Reeh-Schlieder property with respect to } \mathcal{A},\\
    &\implies \omega_\psi \text{ does not have the Verch-Werner ppt property} \implies  \omega_\psi \text{ is entangled.} 
    \end{aligned}
\end{equation}
\end{lem}

The Reeh-Schlieder property has also implications on the mixedness of restricted states.

\begin{lem}
Let $\mathcal{A}, \mathcal{B} \subseteq \mathcal{R}$ be as in the previous lemma. If $\psi$ is a cyclic unit vector for $\mathcal{A}$, then $\omega_\psi \restriction \mathcal{B}$ is a mixed state.
\label{lem_general_restriction_mixed}
\end{lem}

\begin{proof}
We follow an argument in~\cite{Clifton_2001}. 
The proof consists of three parts. 

Firstly, we show that there exists a unit vector $y \in \mathcal{H}$ such that $\omega_y \restriction \mathcal{B} \neq \omega_\psi \restriction\mathcal{B}$: {we proceed by contradiction. Suppose that for every unit vector $y \in \mathcal{H}$ it holds that $\omega_y \restriction \mathcal{B} = \omega_\psi \restriction\mathcal{B}$. For every self-adjoint $B \in \mathcal{B}$, let us set $X_B:= B - \braket{\psi|B\psi} \openone \in \mathcal{B}$, then for every unit $y \in \mathcal{H}:$ $\braket{y|X_By}= \omega_y(B) - \omega_\psi(B)=0$. From this and the fact that for self-adjoint $X_B$, $\|X_B\|= \sup_{\|y\|_\mathcal{H}=1} |\braket{y|X_By}|$ we get that $X_B=0$. In particular, every self-adjoint element in $\mathcal{B}$ is a multiple of the identity. Since every operator in $\mathcal{B}$ can be written as the sum of two self-adjoint operators, we have that every operator in $\mathcal{B}$ is a multiple of the identity, which contradicts the assumption that $\mathcal{B}$ is non-Abelian.}

Secondly, since $\psi$ is cyclic for $\mathcal{A}$, we find ${(A_n)_{n \in \mathbb{N}} \subseteq} \mathcal{A}$ such that $A_n \psi \to y$. Upon restricting to a subsequence and rescaling $A_n$, the sequence $y_n := A_n \psi$ consists of unit vectors and converges to $y$, in particular $\omega_{y_n} \to \omega_y$ pointwise, i.e., in the weak$^*$-topology. If $\omega_y \restriction \mathcal{B} \neq \omega_\psi \restriction\mathcal{B}$, this shows that there exists $n \in \mathbb{N}$ such that $\omega_{y_n} \restriction \mathcal{B} \neq \omega_\psi \restriction\mathcal{B}$.

Thirdly, we show that there exists $\lambda \in (0,1)$ and a state $\tau$ such that  $\omega_\psi \restriction \mathcal{B} = \lambda \omega_{y_n} \restriction \mathcal{B} + (1-\lambda) \tau$. We use that {for all $B \in \mathcal{B}$}
\begin{equation}
    \omega_{y_n}(B^\dagger B) = \omega_\psi(A_n^\dagger B^\dagger B A_n) = \omega_\psi(B^\dagger A_n^\dagger A_n B) \leq \|A_n\|^2 \omega_\psi(B^\dagger B).
\end{equation} 
We set $\lambda:= \frac{1}{ \|A_n\|^2}$, which (by setting $B=\openone$) is in $(0,1]$. To see that $\lambda \neq 1$, note that for two states $\omega_1$ and $\omega_2$ such that $\omega_1(B^\dagger B) \leq \omega_2(B^\dagger B)$, it follows by the Cauchy-Schwarz inequality (see Proposition~5 in~\cite{Baer2009}) applied to the positive functional $\omega_2 - \omega_1$ (which maps $\openone$ to $0$), that $\omega_1 = \omega_2$. As a result we see that $\tau:= \frac{1}{1-\lambda}\qty( \omega_\psi \restriction \mathcal{B} - \lambda \omega_{y_n} \restriction \mathcal{B})$ is a state, which finishes the proof.
\end{proof}

This has an immediate consequence for states restricted to local modes.

\begin{lem}
Let $\mathcal{A}$ be the CCR-$C^*$-theory of a linear scalar field on a globally hyperbolic spacetime $M$ and let $\mathcal{B}$ be the CCR-$C^*$-algebra generated by one mode localisable in a region $N$ such that there exists a region $L \subseteq N^\perp$. Then every state $\omega$ on $\mathcal{A}(M)$ that has the Reeh-Schlieder property with respect to $\mathcal{A}(L)$ restricts to a mixed state on $\mathcal{B}$.
\label{lem_restriction_mixed}
\end{lem}

Remark: The existence of a region $L \subseteq N^\perp$ is stronger than demanding that $N^\perp$ is non-empty {and equivalent to saying that the open interior of $N^\perp$ is non-empty}\footnote{Take for $N:=D(\Sigma\setminus \{p\})$ the domain of dependence of a Cauchy surface $\Sigma$ with one point $p$ removed. Then $N^\perp = \{p\}$ is non-empty but does not contain a region. {However, if the open interior of $N^\perp$ is non-empty, then it contains a region since $M$ is strongly causal~\cite{pfaeffle2009}.}}. For precompact $N$, $\qty(\overline{N})^\perp \subseteq N^\perp $ is a region if it is non-empty.

\begin{proof}
By Einstein causality, $\mathcal{A}(L)$ and $\mathcal{B}$ are two commuting and certainly non-Abelian subalgebras of $\mathcal{A}(M)$. For globally hyperbolic $M$ and normally hyperbolic equations of motion, $\mathcal{A}(M)$ is simple (see Corollary~4.2.10 in~\cite{Baer2007wave}), i.e. it has no non-trivial closed two-sided $*$-ideals, so in particular, $\pi_\omega$ is injective (as it is not the zero representation) and $\pi_\omega\qty[\mathcal{A}(L)]$ and $\pi_\omega\qty[\mathcal{B}]$ are two commuting, non-Abelian, {unital} $C^*$-subalgebras of $BL(\mathcal{H}_\omega)$ and $\Omega_\omega$ is cyclic for $\pi_\omega\qty[\mathcal{A}(L)]$. Hence according to Lemma~\ref{lem_general_restriction_mixed}, $\braket{\Omega_\omega| \cdot \Omega_\omega}$ is a mixed state on $\pi_\omega\qty[\mathcal{B}]$ and hence also a mixed state on $\mathcal{B}$.
\end{proof}

We see in particular, that a Reeh-Schlieder state $\omega$, i.e., one that has the Reeh-Schlieder property with respect to every region, restricts to a truly mixed state on every local mode\footnote{That states of the free scalar field restrict to mixed states on the \emph{full} algebra of precompact double cones of Minkowski spacetime holds more generally, see Corollary 3.3 in~\cite{fewster2013hadamard}.}. {An alternative to the two previous explicit proofs, which show how mixedness follows from the Reeh-Schlieder property directly, is to slightly generalise the statements and argue that restrictions of \emph{all} entangled states are mixed, see Lemma~11.3.6.~in~\cite{KR1997VolII}.}

The significance of this result is reinforced by the fact, that many \emph{physically reasonable} states are Reeh-Schlieder with respect to every region, {for example} quasi-free states of {Klein-Gordon fields} on real analytic spacetimes that fulfill the \emph{analytic microlocal spectrum condition} ($a\mu SC$)~\cite{Strohmaier_2002}; {see~\cite{Gerard2019} for a general existence result and~\cite{Wrochna2020} for the existence of such KMS states on stationary real analytic spacetimes. {It turns out that} all} quasi-free ground and also KMS states on stationary real analytic spacetimes are known to fulfill the $a \mu SC$~\cite{Strohmaier_2002} {and are hence Reeh-Schlieder states}. {In fact, quasi-free ground and KMS states on} general ultrastatic~\cite{Verch1993} and stationary spacetimes~\cite{Strohmaier2000} {are Reeh-Schlieder}. Additionally, the existence of physically reasonable states that have the Reeh-Schlieder property for possibly only a few regions $L$ (which would be sufficient for the above result) on globally hyperbolic spacetimes was analysed in~\cite{Sanders_2009}, see also Sec.~3.2 in~\cite{Fewster2015split}. {On} flat Minkowski spacetime we even have that every state with bounded energy has the Reeh-Schlieder property for every region, see Sec.~II.5.3 in~\cite{haag2012local}.

The result above shows in particular, that quasi-free Reeh-Schlieder states never restrict to pure states on the CCR-$C^*$-algebra of local modes and hence also not to a corresponding ``ground state''. If $a, a^\dagger$ are {some} local annihilation and creation operators of the local mode under consideration and $\sigma$ is a quasi-free Reeh-Schlieder state, then
\begin{equation}
    \begin{aligned}
    \sigma(a^\dagger a) \neq 0.
    \end{aligned}
    \label{eq_vev}
\end{equation}
We briefly mention that this fact is very well-known (and can be proven in a much more direct way from the Reeh-Schlieder property). In the algebraic approach to scattering theory on Minkowski spacetime in~\cite{Araki1967} for instance, one circumvents this by using ``quasi-local'' operators (which cannot be localised in a finite region but might have vanishing vacuum expectation value) instead to define the notion of an asymptotic particle state (see also the discussion of collision theory in Part VI.~of~\cite{haag2012local}). However, since we want to focus on \emph{local} observables, we interpret Eq.~\eqref{eq_vev} as the incarnation of local ``fluctuations'' of quantum fields in physically reasonable states, see also Corollary 27 in~\cite{fewster2019algebraic}.

In the next subsection we will show how we can view such a single local mode of a probe field as \emph{local} particle detector for a coupled system field of interest. In this scenario, the state $\sigma$ takes on the role of the initial state of the probe field and our previous discussion shows that it is impossible to prepare a quasi-free Reeh-Schlieder state $\sigma$ and a local mode of the probe field in such a way that $\hat{\sigma}$ restricts to a pure state. This is a fundamental difference to the singularly or non-locally coupled mode of a non-relativistic harmonic oscillator, which is usually considered to initially be in its pure ground state.

Let us now have a look at the class of models of~\cite{fewster2018quantum} and restrict our attention to a single local mode of the probe field.

\subsection{Bilinearly coupled scalar probe and system fields}

Let us consider a system given by a linear real scalar field $\psi$ with mass $m_S$ and normally hyperbolic eom-operator $P$ on a globally hyperbolic spacetime $M$ and let our local probe be described by another linear real scalar field $\varphi$ with mass $m$ and eom-operator $Q$. The bilinear coupling of the probe and the system we want to consider is the one introduced in Sec.~4 in~\cite{fewster2018quantum} and is expressed in terms of an interaction Lagrangian density by $- \lambda \rho \psi \varphi$, where $\lambda$ is the coupling constant and $\rho$ is a smooth{, real-valued} coupling function with support in some compact coupling zone $K$.

To give an explicit example, we could look at the following Lagrangian density of the coupled model
\begin{equation}
    \begin{aligned}
    \mathcal{L}= &\frac{1}{2} (\nabla_\mu \psi)(\nabla^\mu \psi) - \frac{m_S^2}{2} \psi^2+\frac{1}{2} (\nabla_\mu \varphi)(\nabla^\mu \varphi) - \frac{m^2}{2} \varphi^2 - \lambda \rho \psi \varphi.
    \end{aligned}
\end{equation}
In this case the free eom are given by the Klein-Gordon operators $P=\Box + m_S^2$ and $Q=\Box + m^2$ respectively. While $P=\Box + m_S^2$ might be considered reasonable to model the free dynamics of a system field of interest, one could argue that this is not necessarily the case for $Q=\Box + m^2$; a probe field in a laboratory is usually not a Klein Gordon field. However, we emphasise that the results of our paper are not tied to this specific choice. In particular, we could choose the probe field $\varphi$ to be a massive linear real scalar field under the influence of an \emph{external} field as long as the emerging eom is normally hyperbolic. Such an external field $\chi$ could potentially be used to model a cavity or trap for probe field degrees of freedom by adding a term $- \frac{\chi}{2} \varphi^2$ to $\mathcal{L}$, see also the discussion in Sec.~\ref{Sec_outlook}.

Regardless of the precise form of $P$ and $Q$, the explicit bilinear coupling term allows us to write the coupled eom-operator $T$ on $C^\infty(M;\mathbb{C}) \oplus C^\infty(M;\mathbb{C})$ as
\begin{equation}
    T= \mqty(P & R\\ R &Q),
\end{equation}
where $R$ is the operator of pointwise multiplication with $\rho$. This defines a linear, normally hyperbolic equation of motion and a coupled interacting theory in the sense of the FV framework~\cite{fewster2018quantum}. The associated scattering map $\Theta$ was derived in Appendix D of~\cite{fewster2018quantum}. It acts on the field $\Xi$ of the free combination of system and probe given by $\Xi {\footnotesize \mqty( f \\ h)} := \psi(f) \otimes \openone + \openone \otimes \varphi(h)$. The image of $\Xi {\footnotesize \mqty( f \\ h)}$ under $\Theta$ for $f,g$ with support in $M^+ = M\setminus J^-(K)$ can be written as $\Xi {\footnotesize \qty(\theta \mqty(f \\ h))}$, where
\begin{equation}
    \begin{aligned}
    \theta \mqty(f \\ h) = \mqty(f \\ h) - \mqty(0 & R \\ R& 0) E_T^- \mqty(f \\ h).
    \end{aligned}
\end{equation}
$E_T^-$ is the \emph{advanced} Green operator of $T$ (it fulfills $\mathrm{supp }\qty(E^-_T f) \subseteq J^-(\mathrm{supp}\qty(f))$ for $f \in C_c^\infty(M;\mathbb{C})$). Note that this is a fully non-perturbative expression that holds for arbitrary $\lambda$, however, it is certainly possible to expand $E_T^-$ into a Born-series for small couplings $\lambda$, see Sec.~5.3 in~\cite{fewster2018quantum} and it is even possible to prove stronger analyticity results~\cite{fewster-non-local}. One can also directly check (at least at zeroth and first order in $\lambda$ in renormalised perturbation theory) that
\begin{equation}
    \begin{aligned}
   \Theta\qty(\openone \otimes \varphi(h)) =\mathsf{S}^\dagger \star \qty(\openone \otimes \varphi(h)) \star \mathsf{S},
    \end{aligned}
\end{equation}
where ${\sf S}= \mathcal{T} \exp (\frac{ -\mathrm{i} \lambda}{\hbar} \int\limits_M \rho(x) \psi(x) \varphi(x) \; \mathrm{d}V_M(x))$ is the usual perturbative scattering operator given by a renormalised time-ordered exponential and $\star$ is the star-product \emph{\`a la} deformation quantisation, see for instance~\cite{rejzner2016}.

Let us now consider the situation where an agent couples the probe field $\varphi$ to the system in compact coupling zone $K$ and then analyses the probe in a ``processing region'' $N$, which we reasonably choose to lie outside of the causal past of $K$, i.e. $N \subseteq M^+$. Let us assume that the agent only accesses one single mode localisable in $N$, defined by $h_1,h_2 \in C_c^\infty(N;\mathbb{R})$. As before set $h:= \frac{1}{\sqrt{2}}(h_1 + \mathrm{i} h_2)$. Let us look at the system observable induced by the number operator $a^\dagger a = \varphi(\overline{h}) \varphi(h)$. If the probe is in initial quasi-free state $\sigma$ (with vanishing one-point function), then it is given by (see Eq.~(5.30) in~\cite{fewster2018quantum})
\begin{equation}
    \varepsilon\qty(\varphi\qty(\overline{h}) \varphi(h)) =  \psi\qty(\overline{f^-})\psi\qty(f^-) +  \sigma\qty(\varphi\qty(\overline{h^-}) \varphi\qty(h^-)) \openone,
\end{equation}
where ${\footnotesize\mqty(f^- \\ h^-)}:= \theta {\footnotesize\mqty(0\\h)}$. {Note that both $f^-$ and $h^-$ are smooth, real-valued functions with compact support.} It is tempting to interpret the first term on the right hand side as (possibly {up to an offset} a multiple of) the number operator associated to the {local} system mode\footnote{Generically $E_P(f_1^-,f_2^-) \neq 1 \neq E_Q(h_1^-,h_2^-)$, but also $E_P(f_1^-,f_2^-) \neq 0 \neq E_Q(h_1^-,h_2^-)$. In particular $\qty[\psi(f_1^-),\psi(f_2^-)]\neq \openone$, but they still generate a local mode.} {generated by} $\psi(f^-_1), \psi(f^-_2)$, where ${\footnotesize\mqty(f_j^- \\ h_j^-)}:= \theta {\footnotesize\mqty(0\\h_j)}$ for $j \in \{1,2\}$. 

Recall that for initial system state $\omega$, the result of the agents measurement of the probe observable $a^\dagger a$ is then due to the coupling given by
\begin{equation}
\begin{aligned}
 &\omega(\varepsilon(a^\dagger a)) = \omega \qty(\psi\qty(\overline{f^-})\psi\qty(f^-)) + \sigma\qty(\varphi\qty(\overline{h^-}) \varphi\qty(h^-)).
\end{aligned}
\end{equation}
According to our earlier discussion it follows immediately that the result of our agent's measurement of the probe observable $a^\dagger a$ is ({morally}) the expected number of ``particles'' in {the local} system mode {generated by} $\psi(f_1^-), \psi(f_2^-)$ in the system state $\omega$ \emph{plus} the expected number of ``particles'' in {the local} probe mode {generated by} $\varphi(h_1^-), \varphi(h_2^-)$ in the probe state $\sigma$. If the agent's objective is to learn something about the system, then the second summand is clearly a disturbance in the signal and may be ``attributed to fluctuations in the probe''~\cite{fewster2018quantum} in agreement with our previous discussion. There we argued that this additional fluctuation is {always present} for Reeh-Schlieder states $\sigma$.

Knowing that for such probe preparation states ``noise'' in our agent's signal can never be completely avoided, it is all the more important to ask how it can be reduced. We want to give an outlook for a possible answer to this question via a perturbative analysis in the small coupling regime, i.e., for small $\lambda$. It was shown in Sec.~5.3 in~\cite{fewster2018quantum} that $h^-= h + \lambda^2 R E^-_P R E_Q h +    \mathcal{O}(\lambda^4)$ and $f^-= -\lambda R E_Q^- h + \mathcal{O}(\lambda^3)$, where $E^-_P$, $E^-_Q$ are the advanced Green operators associated to the normally hyperbolic equations of motion of the system and the probe respectively. We then see that
\begin{equation}
\begin{aligned}
 \omega(\varepsilon(a^\dagger a)) =\sigma\qty(a^\dagger a) + \lambda^2& \Big( \sigma\qty(\varphi\qty(\overline{h}) \varphi\qty(R E_P^- R E_Q^- h))+ \sigma\qty(\varphi\qty(R E_P^- R E_Q^- \overline{h}) \varphi\qty(h))\\
 &+\omega \qty(\psi\qty(R E_Q^-\overline{f})\psi\qty(R E_Q^-f)) \Big) + \mathcal{O}(\lambda^4),
\end{aligned}
\end{equation}
so the $\lambda$-dominant contribution to the signal in the detector comes from the detector itself given by $\sigma\qty(a^\dagger a)=\sigma\qty(\varphi\qty(\overline{h}) \varphi\qty(h))$. The signal from the system is only visible in the subleading term. It is hence desirable to make the constant term as small as possible even though one cannot make it vanish for Reeh-Schlieder states $\sigma$. One could attempt this by means of tuning the initial probe state $\sigma$ and the local mode {defined by} $h_1, h_2$. We interpret this as the fact that any measurement requires careful preparation of the probe.

Before we continue let us also quickly comment on the question {which local} system modes (up to normalisation) {determined by} $f_1, f_2 \in C_c^\infty(M;\mathbb{R})$ can be measured with our local detector model. It is shown by the author and collaborators in~\cite{fewster-range} that for every $f_1 \in C_c^\infty(M;\mathbb{R})$ one can find $h_1$ and $\rho$ such that $f_1^-$ is as close to $f_1$ (or an equivalent test function) as desired. However, it is not clear if one can find $h_1, h_2$ and $\rho$ such that $f_1^-$ approximates $f_1$ \emph{and} $f_2^-$ approximates $f_2$.

\section{Entanglement harvesting of the linear scalar field with local modes}
\label{Sec_E_H_local}

After the general, model-independent investigation in Sec.~\ref{Sec_Model_indep_E_H} and the discussion of local particle detectors in the previous section, we now consider entanglement harvesting with two local modes. The probes as well as the system will be given by linear real scalar fields. We will denote the probe fields by $\varphi_\mathsf{A}, \varphi_\mathsf{B}$. Before we discuss the coupled theory of the system and the two probe fields, let us first discuss entanglement.

\subsection{Entanglement of quasi-free states of two modes}

We define a mode for each probe field $\varphi_\mathsf{I}$ by the procedure above: let $f^\mathsf{A}_1, f^\mathsf{A}_2 \in C_c^\infty(N_\mathsf{A};\mathbb{R})$ define agent $\mathsf{A}$'s local mode $\varphi_\mathsf{A}(f^\mathsf{A}_1), \varphi_\mathsf{A}(f^\mathsf{A}_2)$ and similarly for $\varphi_\mathsf{B}$ and $f^\mathsf{B}_1,f^\mathsf{B}_2 \in C_c^\infty(N_\mathsf{B};\mathbb{R})$ {for regions $N_\mathsf{A}, N_\mathsf{B}$, whose causal complements each contain some region}. So we get a single local mode for each probe. This is very convenient, because the entanglement of quasi-free states on the combination of two modes is very well understood as we shall recall in this subsection. Let us investigate a quasi-free state $\sigma'$ on the tensor product of the two modes.

We start by defining matrices
\begin{equation}
    \begin{aligned}
    A_{jk}&:= \sigma'(\qty{\varphi_\mathsf{A}(f^\mathsf{A}_j), \varphi_\mathsf{A}(f^\mathsf{A}_k)}) - 2 \sigma'(\varphi_\mathsf{A}(f^\mathsf{A}_j)) \sigma'( \varphi_\mathsf{A}(f^\mathsf{A}_k)),\\ 
    B_{jk}&:= \sigma'(\qty{\varphi_\mathsf{B}(f^\mathsf{B}_j), \varphi_\mathsf{B}(f^\mathsf{B}_k)}) -  2 \sigma'(\varphi_\mathsf{B}(f^\mathsf{B}_j)) \sigma'( \varphi_\mathsf{B}(f^\mathsf{B}_k)),\\
    C_{jk}&:=\sigma'(\qty{\varphi_\mathsf{A}(f^\mathsf{A}_j) ,\varphi_\mathsf{B}(f^\mathsf{B}_k)}) - 2 \sigma'(\varphi_\mathsf{A}(f^\mathsf{A}_j)) \sigma'( \varphi_\mathsf{B}(f^\mathsf{B}_k)) \\
    &= 2 \sigma'(\varphi_\mathsf{A}(f^\mathsf{A}_j)\varphi_\mathsf{B}(f^\mathsf{B}_k))- 2 \sigma'(\varphi_\mathsf{A}(f^\mathsf{A}_j)) \sigma'( \varphi_\mathsf{B}(f^\mathsf{B}_k)),
    \end{aligned}
\end{equation}
and combine them into the \emph{covariance matrix} $\gamma$, where
\begin{equation}
    \begin{aligned}
    \gamma_{jk}:= \mqty(A&C\\ C^T & B)_{jk}.
    \end{aligned}
\end{equation}
We can then define
\begin{equation}
    \begin{aligned}
    \Omega_{jk}:&={\footnotesize\mqty(\mqty{E_\mathsf{A}(f^\mathsf{A}_1,f^\mathsf{A}_1) & E_\mathsf{A}(f^\mathsf{A}_1,f^\mathsf{A}_2) \\ E_\mathsf{A}(f^\mathsf{A}_2,f^\mathsf{A}_1)& E_\mathsf{A}(f^\mathsf{A}_2,f^\mathsf{A}_2)} & \mbox{\normalfont\Large\bfseries 0} \\ \mbox{\normalfont\Large\bfseries 0}& \mqty{E_\mathsf{B}(f^\mathsf{B}_1,f^\mathsf{B}_1) & E_\mathsf{B}(f^\mathsf{B}_1,f^\mathsf{B}_2) \\ E_\mathsf{B}(f^\mathsf{B}_2,f^\mathsf{B}_1)& E_\mathsf{B}(f^\mathsf{B}_2,f^\mathsf{B}_2)})_{jk}=\mqty(\mqty{0 & 1 \\ -1& 0} & \mbox{\normalfont\Large\bfseries 0} \\ \mbox{\normalfont\Large\bfseries 0}& \mqty{0 & 1 \\ -1& 0})_{jk}}
    \end{aligned}
\end{equation}
so $\Omega_{jk}= \qty(s \oplus s)_{jk}$ where $s_{lm}:={\footnotesize\mqty(0 & 1 \\ -1& 0)_{lm}}$. Note that $s^T = - s$, i.e., the transposition is equivalent to the multiplication by $-1$.
The positivity of the restricted state is equivalent to
\begin{equation}
    \begin{aligned}
    \gamma + \mathrm{i} \Omega \geq 0,
    \end{aligned}
    \label{eq_uncertainty_two_modes}
\end{equation}
the uncertainty relation for the restriction of $\sigma'$, which \emph{always} holds. The following condition, however, only holds for classically correlated quasi-free states:
\begin{equation}
    \begin{aligned}
    \gamma + \mathrm{i} \qty(s\oplus (-s)) \geq 0.
    \end{aligned}
    \label{eq_peres_horodecki}
\end{equation}
This is the Peres-Horodecki ppt condition for the second moments and was investigated in~\cite{Simon_2000}. The following lemma formalises this statement and also sheds light on the relationship of Eq.~\eqref{eq_peres_horodecki} with the Verch-Werner ppt property. The proof can be found in~\ref{sec_appendix_proof_lem_entanglement}.
\begin{lem}
Let $\omega$ be a quasi-free state on two modes with covariance matrix $\gamma$, then the following are equivalent
\begin{enumerate}
    \item $\hat{\omega}$ is \emph{not} Verch-Werner ppt,
    \item Eq.~\eqref{eq_peres_horodecki} does \emph{not} hold,
    \item $\hat{\omega}$ is entangled,
    \item $\omega$ is entangled,
\end{enumerate}
where $\hat{\omega}$ is the quasi-free state on the CCR-$C^*$-algebra of the two modes associated to $\omega$\footnote{The CCR-$C^*$-algebra of the two modes is the (up to isomorphism) \emph{unique} $C^*$-tensor product of the CCR-$C^*$-algebras of the two modes, see for instance~\cite{fewster-range} for details.}.
\label{lem_entangled_bipartite_states}
\end{lem}
Equation~\eqref{eq_peres_horodecki} can be cast in the form~\cite{Simon_2000}
\begin{equation}
\begin{aligned}
       \det(A) + \det(B) - \det (A) \det(B) +\tr (A s C s B s C^T s) -\qty(1 +\det(C))^2 \leq 0,
    \end{aligned}
\end{equation}
which we will call the \emph{Simon} condition (after the author of~\cite{Simon_2000}).
It is worth noting that the uncertainty relation Eq.~\eqref{eq_uncertainty_two_modes} is equivalent to
\begin{equation}
\begin{aligned}
       \det(A) + \det(B) - \det (A) \det(B) +\tr (A s C s B s C^T s) -\qty(1 -\det(C))^2 \leq 0,
    \end{aligned}
\end{equation}

The idea is now the following: We couple the two probe fields to the system. We choose processing regions $N_\mathsf{A}, N_\mathsf{B}$ and local modes. The updated state of the two probes restricted to the combination of these local modes is given by $\sigma'$. This restriction is entangled if and only if it fails to have the ppt property, so we can fully characterise possible entanglement harvesting by two local modes of the two probe fields. We remark that if $\sigma'$ is \emph{not} entangled on the two local modes, this does \emph{not} imply that the full updated state of the two probes is classically correlated (or uncorrelated). It only means that any possible entanglement cannot be detected by the chosen local modes.

\subsection{Interacting structure}

The bilinear coupling between the full probes and the system we want to consider is again the one described in~\cite{fewster2018quantum}, see also Appendix D therein. Let us denote the system field by $\psi$ and the full probe fields as before by $\varphi_\mathsf{A}$ and $\varphi_\mathsf{B}$ respectively. Then the coupling of the probes to the system may be given in terms of the interaction Lagrangian density $-\lambda \rho_\mathsf{A} \psi \varphi_\mathsf{A} -\lambda \rho_\mathsf{B} \psi \varphi_\mathsf{B}$, where $\lambda$ is a common coupling constant and $\rho_\mathsf{J}$ are smooth, {real-valued,} compactly supported coupling functions with support in $K_\mathsf{A}$ and $K_\mathsf{B}$ respectively. The associated scattering map is denoted by $\Theta$.

Let us take a quasi-free state $\omega$ of the full system theory and let $\sigma_\mathsf{A}, \sigma_\mathsf{B}$ be quasi-free states on the full probe fields. Let the symmetric parts of the truncated two-point functions be given by $\frac{1}{2}\beta_\mathsf{S}, \frac{1}{2}\beta_\mathsf{A}, \frac{1}{2}\beta_\mathsf{B}$ respectively. (The terminology of this section is summarised in Table~\ref{tbl_overview}.)

\begin{table}[h!]
\centering
\renewcommand{\arraystretch}{1.}
\begin{tabular}{ |c||c|c|c|  }
 \hline
 & System & Agent $\mathsf{A}$ & Agent $\mathsf{B}$ \\ [0.5ex]
 \hline\hline
 field                      & $\psi$    & $\varphi_\mathsf{A}$  & $\varphi_\mathsf{B}$\\ \hline
 normally hyperbolic free eom & $P$     & $Q_\mathsf{A}$        & $Q_\mathsf{B}$\\ \hline
 advanced Green operator & $E^-$     & $E^-_\mathsf{A}$        & $E^-_\mathsf{B}$\\ \hline
 initial quasi-free state   & $\omega$  & $\sigma_\mathsf{A}$    & $\sigma_\mathsf{B}$\\ \hline
 $\mqty{\text{symmetric part of} \\ \text{truncated 2-point function}}$ & $\frac{1}{2} \beta_\mathsf{S}$ & $\frac{1}{2} \beta_\mathsf{A}$& $\frac{1}{2} \beta_\mathsf{B}$\\ \hline
 compact coupling zone          &           & $K_\mathsf{A}$        & $K_\mathsf{B}$\\\hline
 interaction Lagrangian density &     & $- \lambda \rho_\mathsf{A} \varphi_\mathsf{A} \psi$        & $- \lambda \rho_\mathsf{B} \varphi_\mathsf{B} \psi$\\ \hline
 final quasi-free state        & & \multicolumn{2}{|c|}{$(\omega \otimes \sigma_\mathsf{A} \otimes \sigma_\mathsf{B})(\Theta (\openone \otimes \cdot))$}\\ \hline
 processing region          &           & $N_\mathsf{A}$        & $N_\mathsf{B}$\\ \hline
 local mode         &           & $\varphi_\mathsf{A}(f^\mathsf{A}_1)$, $\varphi_\mathsf{A}(f^\mathsf{A}_2)$  & $\varphi_\mathsf{B}(f^\mathsf{B}_1)$, $\varphi_\mathsf{B}(f^\mathsf{B}_2)$\\ \hline
 final quasi-free state of the two modes       & & \multicolumn{2}{|c|}{$\sigma'$}\\ \hline
 \hline
\end{tabular}
\caption{Summary of terminology and notation.}
\label{tbl_overview}
\end{table}

The symmetric part of the truncated two-point function of $\omega \otimes \sigma_\mathsf{A} \otimes \sigma_\mathsf{B}$ is given by the direct sum $\frac{1}{2} \beta := \frac{1}{2} (\beta_\mathsf{S} \oplus \beta_\mathsf{A} \oplus \beta_\mathsf{B})$. Let us define $\sigma'$ to be the restriction of $(\omega \otimes \sigma)(\Theta(1\!\!\!\!1 \otimes \cdot))$ to the combination of the two local modes under consideration. We see that $\sigma'$ is quasi-free as well, in particular for $i,j,k,l \in \{1,2\}$
\begin{equation}
    \begin{aligned}
    &\sigma'\qty(\qty{\varphi_{\mathsf{A},\mathsf{B}} {\footnotesize\mqty(f^\mathsf{A}_i \\ f^\mathsf{B}_j )}, \varphi_{\mathsf{A},\mathsf{B}}  {\footnotesize \mqty(f^\mathsf{A}_k \\ f^\mathsf{B}_l )}}) - 2 \sigma'\qty(\varphi_{\mathsf{A},\mathsf{B}} {\footnotesize\mqty(f^\mathsf{A}_i \\ f^\mathsf{B}_j )}) \sigma'\qty( \varphi_{\mathsf{A},\mathsf{B}}  {\footnotesize\mqty(f^\mathsf{A}_k \\ f^\mathsf{B}_l )})= \beta \qty(\theta {\footnotesize\mqty(0\\f^\mathsf{A}_i\\f^\mathsf{B}_j)}, \theta{\footnotesize\mqty(0\\f^\mathsf{A}_k\\f^\mathsf{B}_l)}),
    \end{aligned}
\end{equation}
where $\varphi_{\mathsf{A},\mathsf{B}} {\footnotesize\mqty(f^\mathsf{A}_i \\ f^\mathsf{B}_j )}:= \varphi_\mathsf{A}(f^\mathsf{A}_i) \otimes \openone + \openone \otimes \varphi_\mathsf{B}(f^\mathsf{B}_j)$, see also {the last line of} Eq.~\eqref{eq_two-point_func}. As before, $\theta$ is defined such that 
\begin{equation}
    \begin{aligned}
    \Theta \qty(\Xi \qty(\vec{g}) ) =  \Xi \qty(\theta\vec{g}),
    \end{aligned}
\end{equation}
for $\Xi {\footnotesize\qty(\vec{g})} := \psi (g^0) \otimes \openone + \openone \otimes \varphi_{\mathsf{A},\mathsf{B}} {\footnotesize\mqty(g^1 \\ g^2 )}$, with $\vec{g}$ supported in $M \setminus J^-(K_\mathsf{A} \cup K_\mathsf{B})$. Let us introduce the following notation

\begin{equation}
    \begin{aligned}
    \mqty(F_j^{\mathsf{A},{\mathsf S}}\\ F_j^{\mathsf{A},{\mathsf A}} \\F_j^{\mathsf{A},{\mathsf B}}) &:=\theta \mqty(0\\ f^\mathsf{A}_j \\ 0) , \qquad \mqty(F_j^{\mathsf{B},{\mathsf S}}\\ F_j^{\mathsf{B},{\mathsf A}} \\F_j^{\mathsf{B},{\mathsf B}}) &:=\theta \mqty(0\\ 0 \\ f^\mathsf{B}_j).
    \end{aligned}
\end{equation}
We then have $A_{jk} =  \sum_l \beta_l(F_j^{\mathsf{A},l},F_k^{\mathsf{A},l})$, $B_{jk} =  \sum_l \beta_l(F_j^{\mathsf{B},l},F_k^{\mathsf{B},l})$ and $C_{jk} =  \sum_l \beta_l(F_j^{\mathsf{A},l},F_k^{\mathsf{B},l})$, where the sums run over $l \in \{{\mathsf S, \mathsf A, \mathsf B}\}$.

Let us now assume that $K_\mathsf{A}$ is spacelike separated from $K_\mathsf{B}$, from which, together with causal factorisation (see Corollary D.2.~in~\cite{fewster2018quantum}), it follows that $\theta=\hat{\theta}_1 \circ \hat{\theta}_2 = \hat{\theta}_2 \circ \hat{\theta}_1$. Then one can see that 
\begin{equation}
    \begin{aligned}
    \theta \mqty(0\\ f^\mathsf{A} \\ 0) &= \hat{\theta}_1 \mqty(0\\ f^\mathsf{A} \\ 0)= \mqty(F^{\mathsf{A},{\mathsf S}}\\ F^{\mathsf{A},{\mathsf A}} \\F^{\mathsf{A},{\mathsf B}}), \qquad
    \theta \mqty(0\\ 0 \\ f^\mathsf{B}) = \hat{\theta}_2 \mqty(0\\ 0 \\ f^\mathsf{B}) = \mqty(F^{\mathsf{B},{\mathsf S}}\\ F^{\mathsf{B},{\mathsf A}} \\F^{\mathsf{B},{\mathsf B}}),
    \end{aligned}
\end{equation}
which means in particular that $F^{\mathsf{A},{\mathsf B}} \equiv 0$ and $F^{\mathsf{B},{\mathsf A}} \equiv 0$, thereby simplifying the above expressions to: 
\begin{equation}
    \begin{aligned}
    A_{jk} &=  \sum\limits_{l\in \{{\mathsf S, \mathsf A}\}} \beta_l(F_j^{\mathsf{A},l},F_k^{\mathsf{A},l}), \quad B_{jk}=  \sum\limits_{l\in \{{\mathsf S, \mathsf B}\}} \beta_l(F_j^{\mathsf{B},l},F_k^{\mathsf{B},l}), \quad
    C_{jk} = \beta_{\mathsf S}(F_j^{\mathsf{A},{\mathsf S}},F_k^{\mathsf{B},{\mathsf S}}).
    \end{aligned}
\end{equation}
Given these expressions for the matrices $A, B, C$ we can now investigate the entanglement of the quasi-free state $\sigma'$ on the two local modes. 

\subsection{Non-perturbative analysis: The critical coupling}

The first observation that we can make at this stage is that for sufficiently regular $\beta$, we see that $A, B, C$ (or rather their entries) are continuous functions of the common coupling strength $\lambda$, so $A(\lambda), B(\lambda), C(\lambda)$, following from the Born-series expansion in Sec.~5.3 in~\cite{fewster2018quantum} at least for small $\lambda$. This means that for
\begin{equation}
\begin{aligned}
    p_S(\lambda):= \det(A) + \det(B)- \det (A) \det(B) +\tr (A s C s B s C^T s)- &\qty(1 +\det(C))^2,
    \end{aligned}
\end{equation}
(where we suppressed the $\lambda$-dependence of $A, B, C$), the \emph{Simon} condition can now be cast in the form
\begin{equation}
    \sigma' \text{ is entangled for coupling strength } \lambda \iff p_S(\lambda) >0.
\end{equation}
for the function $p_S:\mathbb{R} \to \mathbb{R}$, which is continuous at least in a neighbourhood around $\lambda =0$. 

The Simon condition is related to another measure of entanglement. Let us set $\Delta:= \det(A) + \det(B) - 2 \det(C)$ and $I_4:= \det(A) \det(B) + \det(C)^2 - \tr\qty(AsCsBsC^Ts)$, and also
\begin{equation}
    \begin{aligned}
    \nu_-(\lambda) := \sqrt{\frac{\Delta - \sqrt{\Delta^2 - 4I_4}}{2}},
    \end{aligned}
\end{equation}
which is the smallest symplectic eigenvalue of $\tilde{\gamma} = \Lambda{^T} \gamma \Lambda$ (see~\ref{sec_appendix_proof_lem_entanglement}), the covariance matrix associated to the partial transpose {of the} state {$\sigma'$}~\cite{Vidal_2002, Adesso_2004}. Based on $\nu_-$ we can define the \emph{negativity} of $\sigma'$~\cite{Vidal_2002, Adesso_2004} given by
\begin{equation}
    \begin{aligned}
    N(\lambda):= \max \qty{0, \frac{1-\nu_-(\lambda)}{2\nu_-(\lambda)}},
    \end{aligned}
\end{equation}
which is an \emph{entanglement monotone}, it does not increase under \emph{local operations and classical communication} (LOCC)~\cite{Vidal_2002}. The negativity is a very common entanglement measure in the literature about entanglement harvesting, however, we will continue using the \emph{Simon} condition, which is no restriction since~\cite{Vidal_2002, Adesso_2004}
\begin{equation}
    \begin{aligned}
    &\sigma' \text{ is entangled for coupling strength } \lambda\\
    &\iff p_S(\lambda) >0 \iff \nu_-(\lambda) < 1 \iff N(\lambda) > 0.
    \end{aligned}
\end{equation}

Let us investigate the situation $\lambda=0$, i.e., the case of no coupling between system and probes. It is easy to see that $C(0)=0$. We set $A(0)=:A_0$, $B(0)=:B_0$ and get
\begin{equation}
\begin{aligned}
    p_S(0)&= \det(A_0) + \det(B_0)-\det (A_0) \det(B_0)  - 1\\
    &= -(\det (A_0) - 1)(\det (B_0) - 1).
    \end{aligned}
\end{equation}
$A_0$ is obviously the covariance matrix associated to $\sigma_\mathsf{A}$ (when restricted to the single mode of interest) and $B_0$ is the covariance matrix associated to $\sigma_\mathsf{B}$ (also restricted), i.e., $(A_0)_{jk}= \beta_{\mathsf A}(f^\mathsf{A}_j,f^\mathsf{A}_k)$ and $(B_0)_{jk}= \beta_{\mathsf B}(f^\mathsf{B}_j,f^\mathsf{B}_k)$. {According to Eq.~\eqref{eq_positivity_condition_mode} and the comments below it, $\det(A_0) \geq 0$ and $\det(B_0)\geq0$, hence $p_S(0) \leq 0$.}

We can now make the following striking observation.

\begin{theo}
Suppose that $\det (A_0)\neq 1$ and $\det(B_0)\neq 1$, then there exists $\lambda_{\mathrm{min}} >0$ such that $\sigma'$ is \emph{not} entangled for all coupling strengths $\lambda$ with $|\lambda| \leq \lambda_{\mathrm{min}}$.
\end{theo}

\begin{proof}
Under the assumption we have that $p_S(0) <0$. As $p_S$ is continuous at $\lambda=0$, the statement follows. 
\end{proof}

This means that unless one of the terms $\det(A_0), \det(B_0)$ equals $1$, there is always a ``coupling strength threshold for entanglement harvesting'' in the $\lambda$-parameter space, i.e., the coupling needs to reach a certain critical strength before it can possibly induce any entanglement at all. Lemma~\ref{lem_pure_det} shows that $\det (A_0)\neq 1$ is equivalent to $\hat{\sigma}_\mathsf{A}$ {(restricted to the local mode)} being a mixed state and as shown in Lemma~\ref{lem_restriction_mixed}, this is \emph{always} the case for a quasi-free Reeh-Schlieder state $\sigma_\mathsf{A}$.

\begin{corol}
If both probes are prepared in quasi-free Reeh-Schlieder states and if the system is prepared in a quasi-free state, then there exists a critical minimal coupling for entanglement harvesting.
\end{corol}

\subsection{Perturbative analysis}

The existence of the critical coupling forms a potentially serious obstacle for any realistic attempt to harvest entanglement, however, the hope is that the threshold can be made small enough by sufficiently tuning the initial probe states {and by choosing appropriate local modes}. In this section we want to start a perturbative analysis and speculate about conditions for successful entanglement harvesting in the perturbative regime.

Again following~\cite{fewster2018quantum} we have
\begin{equation}
    \begin{aligned}
    F_j^{\mathsf{A},{\mathsf A}}&= f^\mathsf{A}_j + \lambda^2 {U_\mathsf{A}} f^\mathsf{A}_j + {\lambda^4 U_\mathsf{A}^2 f^\mathsf{A}_j +}   \mathcal{O}(\lambda^{6}), F_j^{\mathsf{A},{\mathsf S}}\!= -\lambda R_\mathsf{A} E_{\mathsf{A}}^- f^\mathsf{A}_j {- \lambda^3 R_\mathsf{A} E_{\mathsf{A}}^- U_\mathsf{A} f^\mathsf{A}_j}\!+\! \mathcal{O}(\lambda^{5}),\\
    F_k^{\mathsf{B},{\mathsf B}}&= {f_k^\mathsf{B}} + \lambda^2  {U_\mathsf{B}} f^\mathsf{B}_k + {\lambda^4 U_\mathsf{B}^2 f_k^\mathsf{B}+} \mathcal{O}(\lambda^{6}),
    F_k^{\mathsf{B},{\mathsf S}}\!= -\lambda R_\mathsf{B} E_{\mathsf{B}}^- f^\mathsf{B}_k {- \lambda^3 R_\mathsf{B} E_{\mathsf{B}}^- U_\mathsf{B} f^\mathsf{B}_k}\!+\! \mathcal{O}(\lambda^{5}),
    \end{aligned}
\end{equation}
where, as before, $E^-$, $E^-_\mathsf{A}$ and $E^-_\mathsf{B}$ are the advanced Green operators associated to the normally hyperbolic equations of motion of the system and the two probes respectively, $R_\mathsf{A}$ is the operator that multiplies pointwise with the coupling function $\rho_\mathsf{A}$, ${ U_\mathsf{A}:=R_\mathsf{A} E^- R_\mathsf{A} E_\mathsf{A}^-}$ {and} similar for $R_\mathsf{B}$ {and $U_\mathsf{B}$}. With this we get the expansions

\begin{equation}
    \begin{aligned}
    A_{jk} &= \beta_\mathsf{A}(f^\mathsf{A}_j,f^\mathsf{A}_k) + \lambda^2 \Big(\beta_\mathsf{S}( R_\mathsf{A} E_{\mathsf{A}}^- f^\mathsf{A}_j,R_\mathsf{A} E_{\mathsf{A}}^- f^\mathsf{A}_k)+\beta_\mathsf{A}(f^\mathsf{A}_j,{U_\mathsf{A}} f^\mathsf{A}_k)+\beta_\mathsf{A}({U_\mathsf{A}} f^\mathsf{A}_j,f^\mathsf{A}_k)\Big) \\
    &{ +\lambda^4 \Big( \beta_\mathsf{S}(R_\mathsf{A}E_\mathsf{A}^- f_j^\mathsf{A}, R_\mathsf{A}E_\mathsf{A}^- U_\mathsf{A}f_k^\mathsf{A}) + \beta_\mathsf{S}(R_\mathsf{A}E_\mathsf{A}^- U_\mathsf{A}f_j^\mathsf{A}, R_\mathsf{A}E_\mathsf{A}^- f_k^\mathsf{A})}\\
    &{+\beta_\mathsf{A}(f_j^\mathsf{A}, U_\mathsf{A}^2 f_k^\mathsf{A}) + \beta_\mathsf{A}(U_\mathsf{A}f_j^\mathsf{A}, U_\mathsf{A} f_k^\mathsf{A}) + \beta_\mathsf{A}(U_\mathsf{A}^2f_j^\mathsf{A},  f_k^\mathsf{A})\Big)}+\mathcal{O}(\lambda^{6}),\\
    B_{jk} &= \beta_\mathsf{B}(f^\mathsf{B}_j,f^\mathsf{B}_k) + \lambda^2 \Big(\beta_\mathsf{S}( R_\mathsf{B} E_{\mathsf{B}}^- f^\mathsf{B}_j,R_\mathsf{B} E_{\mathsf{B}}^- f^\mathsf{B}_k)+\beta_\mathsf{B}(f^\mathsf{B}_j,{U_\mathsf{B}} f^\mathsf{B}_k)+\beta_\mathsf{B}({U_\mathsf{B}} f^\mathsf{B}_j,f^\mathsf{B}_k)\Big) \\
    &{ +\lambda^4 \Big( \beta_\mathsf{S}(R_\mathsf{B}E_\mathsf{B}^- f_j^\mathsf{B}, R_\mathsf{B}E_\mathsf{B}^- U_\mathsf{B}f_k^\mathsf{B}) + \beta_\mathsf{S}(R_\mathsf{B}E_\mathsf{b}^- U_\mathsf{B}f_j^\mathsf{B}, R_\mathsf{B}E_\mathsf{B}^- f_k^\mathsf{B})}\\
    &{+\beta_\mathsf{B}(f_j^\mathsf{B}, U_\mathsf{B}^2 f_k^\mathsf{B}) + \beta_\mathsf{B}(U_\mathsf{B}f_j^\mathsf{B}, U_\mathsf{B} f_k^\mathsf{B}) + \beta_\mathsf{B}(U_\mathsf{B}^2f_j^\mathsf{B},  f_k^\mathsf{B})\Big)}+\mathcal{O}(\lambda^{6}),\\
    C_{jk} &= \lambda^2 \beta_\mathsf{S}(R_\mathsf{A} E_{\mathsf{A}}^- f^\mathsf{A}_j,R_\mathsf{B} E_{\mathsf{B}}^- f^\mathsf{B}_k) + \mathcal{O}(\lambda^4).
    \end{aligned}
    \label{eq_perturbative_expansion_ABC}
\end{equation}
Note that the initial probe states $\sigma_\mathsf{A},\sigma_\mathsf{B}$ contribute to $A$ and $B$ respectively through their values { on the local modes in the processing regions $N_\mathsf{A}, N_\mathsf{B}$ themselves, i.e.,} $\sigma_\mathsf{A}(\qty{\varphi_\mathsf{A}(f^\mathsf{A}_j), \varphi_\mathsf{A}(f^\mathsf{A}_k)})$ and $\sigma_\mathsf{B}(\qty{\varphi_\mathsf{B}(f^\mathsf{B}_j), \varphi_\mathsf{B}(f^\mathsf{B}_k)})$, \emph{but also} through { their correlations between the degrees of freedom of the processing region and (a region containing) the coupling zone and within the latter itself via, e.g.,}
\begin{equation}
    \begin{aligned}
    &\sigma_\mathsf{A}(\qty{\varphi_\mathsf{A}( f_j^\mathsf{A}),  \varphi_\mathsf{A}({U_\mathsf{A}} f^\mathsf{A}_k)}),\; \sigma_\mathsf{A}(\qty{\varphi_\mathsf{A}({U_\mathsf{A}}f^\mathsf{A}_j), \varphi_\mathsf{A}({U_\mathsf{A}} f^\mathsf{A}_k)}),\\
    &\sigma_\mathsf{B}(\qty{\varphi_\mathsf{B}( f_j^\mathsf{B}),  \varphi_\mathsf{B}({U_\mathsf{B}} f^\mathsf{B}_k)}),\; \sigma_\mathsf{B}(\qty{\varphi_\mathsf{B}({U_\mathsf{B}}f^\mathsf{B}_j), \varphi_\mathsf{B}({U_\mathsf{B}} f^\mathsf{B}_k)}),
    \end{aligned}
\end{equation} 
in addition to possible one-point function terms. We see that $p_S(\lambda)$ can then be written as a convergent power series in a neighbourhood around $\lambda=0$, so we write $p_S(\lambda) = p_S(0) + \sum_{n=1}^N \lambda^n p_n + \mathcal{O}(\lambda^{N+1})$. It holds that $p_n=0$ for odd $n$.

Let us now compute the coefficients $p_n$. For that we look at the {expansions} of the matrices $A, B, C$ in Eq.~\eqref{eq_perturbative_expansion_ABC} above and define $A_n, B_n, C_n$ to be the coefficients of $\lambda^n$ in the respective {expansions}. We have
\begin{equation}
    \begin{aligned}
    \det(A) =& \det(A_0) + \lambda^2 {\det(A_0) \Tr(A_0^{-1}A_2)} \\
    &+ \lambda^4 {(}\det(A_2) {+ \det(A_0) \Tr(A_0^{-1} A_4))}+ \mathcal{O}(\lambda^{6}),\\
    \det(B) =& \det(B_0) + \lambda^2 {\det(B_0) \Tr(B_0^{-1}B_2)} \\
    &+ \lambda^4 {(}\det(B_2) {+ \det(B_0) \Tr(B_0^{-1} B_4))}+ \mathcal{O}(\lambda^{6}),\\
    \det(C) =& \lambda^4 \det(C_2) + \mathcal{O}(\lambda^{6}),\\
    \Tr(AsCsBsC^Ts) =& \lambda^4 \Tr(A_0sC_2sB_0sC_2^Ts) + \mathcal{O}(\lambda^{6}).
    \label{eq_expansions_dets_tr}
    \end{aligned}
\end{equation}
We {then} find that 
\begin{equation}
    \begin{aligned}
    p_S(\lambda) = &-(\det(A_0)-1)(\det (B_0)-1)\\
    &+ \lambda^2 \Big((1-\det (A_0)) {\det(B_0) \Tr(B_0^{-1} B_2)} +  (1-\det (B_0)) {\det(A_0) \Tr(A_0^{-1} A_2)} \Big)\\
    &+ \lambda^4 \Big({\Tr} (A_0 s C_2 s B_0 s C_{2}^T s) - 2\det(C_2) \\
    &- {\det(A_0)\det(B_0)\Tr(A_0^{-1}A_2)\Tr(B_0^{-1}B_2)}\\ 
    &+ (1-\det(B_0))(\det(A_2) {+\det(A_0) \Tr(A_0^{-1} A_4)}) \\
    &+ (1-\det(A_0))(\det(B_2){+ \det(B_0) \Tr(B_0^{-1} B_4)})\Big)+ \mathcal{O}(\lambda^{6}).
    \end{aligned}
\end{equation}
Let us also recall that the uncertainty relation Eq.~\eqref{eq_uncertainty_two_modes} for $\sigma'$, which \emph{always} has to hold, implies $\forall \lambda: \tilde{p}_S(\lambda) \leq 0$, where
\begin{equation}
\begin{aligned}
    \tilde{p}_S(\lambda):= -\det (A) \det(B)  - &\qty(1 -\det(C))^2 +\tr (A s C s B s C^T s)+ \det(A) + \det(B).
    \end{aligned}
\end{equation}
Comparing the expansions of $p_S$ and $\tilde{p}_S$ shows that $p_0 = \tilde{p}_0$, $p_2 = \tilde{p}_2$ and $p_4 = \tilde{p}_4 - 4 \det(C_2)$. We again see the common theme of the present paper: the dominant contribution $p_0$ comes from the initial states of the probes, which is also responsible for the formation of the coupling strength threshold for entanglement harvesting. Moreover, entanglement harvesting is in general neither a leading nor a subleading order effect in perturbation theory but rather appears at sub-subleading order. This is different from the results obtained using non-local detectors in pure initial states (see for instance~\cite{Pozas-Kerstjens_2015}), where entanglement harvesting is found to be a leading order effect. The discrepancy is explained by noting that for \emph{pure} initial states{, i.e., for $\det(A_0) = \det(B_0) =1$}, both $p_0$ and $p_2$ vanish, rendering $\lambda^4 p_4$ the leading term {(which then also takes on a considerably simpler form)}. This also reveals a tension: a perturbative analysis requires the coupling to be small, yet entanglement harvesting does not occur below the critical coupling. Nevertheless, we hope a numerical investigation of the values of $p_0, p_2, p_4$ could shed some light on the magnitude of the critical coupling and thereby also on the plausibility of entanglement harvesting by local particle detectors in initially mixed quasi-free states.

\section{Discussion and outlook}
\label{Sec_outlook}

A pair of non-relativistic particle detectors (such as the quantum mechanical harmonic oscillator Unruh detector), coupled to a linear real scalar quantum field either singularly (on a worldline) or non-locally (around a worldline) provides a simple mathematical tool to investigate correlations and entanglement of states of the latter via the entanglement harvesting protocol. However, concerns may be raised whether it is valid to interpret such non-local particle detectors as \emph{realistic} physical systems (or approximations thereof) in \emph{every} regime, simply because the underlying classical equations of motion of the combined interacting structure are either singular or non-local and therefore somewhat unphysical from a perspective of fundamental interactions\footnote{It is undoubted that non-local particle detector models are good approximations in a non-relativistic regime; they are simplifications of the non-relativistic light-matter interaction of non-relativistic quantum electrodynamics (Pauli-Fierz Hamiltonian), see for instance Sec.~II.~in~\cite{Martinez2013_UdW_QED} and Sec.~II.~in~\cite{Pozas2016_UdW_hydrogen}.}. In particular, it is unclear if the results on entanglement harvesting by  non-local particle detectors carry over to more realistic situations.

This was our motivation to implement a fully covariant and local model-independent version of the entanglement harvesting protocol using local probes whose consistency we demonstrated. In a specific model, we introduced the notion of a \emph{local} particle detector given by a local mode of a linear real scalar probe field (possibly under the influence of external fields) that is bilinearly coupled to a system field of interest. We thereby avoid any singular or non-local behaviour. Our interpretation is that such a local probe is a  proxy for a physically realistic and in particular \emph{local} measurement device and that a corresponding local particle detector is a choice of a local degree of freedom of said device that is accessible to an experimenter. To support this point let us sketch a possible way of modelling a scalar field in a finitely extended cavity or trap by a linear normally hyperbolic equation of motion involving external fields, to which our results are directly applicable. A cavity may be modelled as a massive cylindrical shell of a certain thickness at every instance of a certain finite period of time. Let the spacetime-dependent external field $\chi$ be proportional to the (smoothed) characteristic function of this cylindrical shell. Its influence on a probe field $\varphi$ can now be described by the normally hyperbolic equation of motion $(\Box +m^2) \varphi + \chi \varphi =0$. For ``large $\chi$'' this may be interpreted as an effective implementation of vanishing Dirichlet boundary conditions of the field posed by the cavity. A local particle detector would then correspond to a local degree of freedom of $\varphi$ localisable in a region that is ``enclosed'' by the cavity. While the cavity or trap itself is described by an effective non-dynamical external field, this formulation allows to maintain the local and relativistic nature of the probe degree of freedom. It would be interesting to formulate and investigate this sketch in more detail.

The motivation for choosing the algebra of a local mode as this local degree of freedom (and not a different algebra spanned by different local observables) is obviously its simplicity and also the consequent analogy with the (harmonic oscillator) Unruh detector. Whether this means that local modes are a \emph{good} choice is a variant of the notoriously difficult question of ``the correspondence between physical apparatus and mathematical objects'', quickly touched upon in~\cite{Araki1967} (see also Part VI. in~\cite{haag2012local}), which unfortunately remains open. 

Continuing with our choice of an accessible local degree of freedom we see, that (by construction) the initial preparation state of such a local particle detector is the restriction of a physically reasonable state of the whole probe \emph{field}. Typical examples of such states are quasi-free, have the Reeh-Schlieder property, and consequently always restrict to \emph{mixed} quasi-free states on local modes.

This is a general fact and has interesting consequences for entanglement harvesting, in particular we found that if the system is in a quasi-free state and if both local modes are initially in an uncorrelated mixed quasi-free state, then there exists a critical coupling strength below which they cannot become entangled. This result is again general and we expect it to likewise apply to Unruh detectors: if they are prepared in a mixed quasi-free state instead of their pure ground state, then they also have a coupling strength threshold for entanglement harvesting. We hence believe that our results motivate an investigation of entanglement harvesting with Unruh detectors that are \emph{not} prepared in their ground states. 

An important question left open is whether reasonably prepared local particle detectors can harvest (not just correlation but) \emph{actual} entanglement. Heuristically, the reason for the mixedness of preparation states of local modes is their entanglement with the other local modes of the probe field. One would then expect that switching on an external field (our theoretical model for installing a cavity) serves not only the purpose of approximately confining certain local modes but also of isolating said modes from the rest of the field outside of the cavity. Furthermore, heuristics lead us to speculate that this allows one to prepare a local mode inside the cavity \emph{very} close to its ground state and thereby also to \emph{vastly} reduce the critical coupling. In theory, it is very easy to add the external field only to the two agents' probe fields, however, in practice, the system field might also ``feel'' the presence of the cavity. In particular, cavities might also isolate the system, which would result in a severe reduction of entanglement of the system observables enclosed by the cavities. Hence there would be practically no entanglement left to be harvested by the confined local particle detectors. This could be avoided by engineering cavities that are effectively transparent for the system field but capable of confining and isolating the probes. Anyway, the first interesting step would be to investigate the magnitude of the critical coupling in very simple cases. Despite the fact that the existence of the critical coupling could be an obstruction to the validity of a perturbative analysis (as hinted at in Sec.~\ref{Sec_E_H_local}), we believe that a numerical investigation could provide some indications in the case of massless Klein-Gordon fields in $1+1$ dimensional Minkowski spacetime.

Finally, it is certainly interesting to go beyond entanglement harvesting and to explore more possible applications of local particle detectors at the interface of relativistic quantum field theory and quantum information and to further assess their physical significance. An application to the Hawking effect is currently work in progress~\cite{Passegger_Hawking}.

\ack
It is a great pleasure to express my gratitude to Henning Bostelmann and Christopher~J.~Fewster for their support, supervision and invaluable input on the content and presentation of this paper, {to Leonardo Garc\'ia Heveling for helpful discussions,} to Yoobin Jeong for helpful comments on the text and to Albert Georg Passegger for useful correspondence. {Furthermore, I owe thanks to one of the anonymous referees for many valuable remarks, in particular, for their suggestion to simplify the first part of the proof of Lemma~\ref{lem_general_restriction_mixed} and for spotting a crucial error in a previous version of Eq.~\eqref{eq_expansions_dets_tr}.} Parts of this work were presented at the ``First Virtual Meeting 2021'' of the fpUK virtual centre as well as at ``RQI-Online 2020/21'' organised by the International Society for Relativistic Quantum Information and I thank the organisers and participants for stimulating discussions and constructive comments. This work was supported by a Mathematics Excellence Programme Studentship awarded by the Department of Mathematics at the University of York.

\appendix

\section{Proof of Theorem~\ref{theo_regions_of_separability}}
\label{sec_appendix_proof_theo_regions_of_sep}

We consider a bipartite probe theory $\mathcal{P}= \mathcal{P}_\mathsf{A} \otimes \mathcal{P}_\mathsf{B}$ coupled to a system theory $\mathcal{S}$ in coupling zone $K = K_\mathsf{A} \cup K_\mathsf{B}$ such that $K_\mathsf{A} \cap J^+(K_\mathsf{B}) = \emptyset$. We have individual scattering maps $\Theta_\mathsf{A}: \mathcal{S} \otimes \mathcal{P}_\mathsf{A} \to \mathcal{S} \otimes \mathcal{P}_\mathsf{A}$ and $\Theta_\mathsf{B}: \mathcal{S} \otimes \mathcal{P}_\mathsf{B} \to \mathcal{S} \otimes \mathcal{P}_\mathsf{B}$ and initial system state $\omega$ and probe state $\sigma = \sigma_\mathsf{A} \otimes \sigma_\mathsf{B}$. Moreover we consider regions $N_\mathsf{A}, N_\mathsf{B}$.

Theorem~\ref{theo_regions_of_separability} follows from the following lemmas.

\begin{lem}[Probe-induced probe observable]
There exists a linear map $\pi_{\omega, \Theta}: \mathcal{P} \to \mathcal{P}$ such that for all states $\sigma$ on $\mathcal{P}:$
\begin{equation}
    \begin{aligned}
    \sigma(\pi_{\omega, \Theta}(\cdot))=(\omega \otimes \sigma)(\Theta (1\!\!\!\!1 \otimes \cdot)).
    \end{aligned}
\end{equation}
Moreover, $\sigma_{\omega, \Theta}:=\sigma \circ \pi_{\omega, \Theta}$ is a state on $\mathcal{P}$.
\end{lem}

\begin{proof}
Let us look at $\eta_\omega: \mathcal{S} \otimes \mathcal{P} \to \mathcal{P}$ defined by the linear extension of $\tilde{\eta}_\omega(A \otimes B) = \omega(A) B$. Then set $\pi_{\omega, \Theta}(\cdot) := \tilde{\eta}_\omega( \Theta(1\!\!\!\!1 \otimes \cdot))$. The rest of the proof follows the similar discussion in Sec.~3.2 in~\cite{fewster2018quantum}.
\end{proof}

Let us make the following observation:
\begin{equation}
\begin{aligned}
\hat{\Theta}_{\sf A} \circ \hat{\Theta}_{\sf B} (\openone \otimes O_{\sf A} \otimes O_{\sf B}) &=\hat{\Theta}_{\sf A} \circ \hat{\Theta}_{\sf B} (\openone \otimes O_{\sf A} \otimes \openone \cdot  \openone \otimes \openone \otimes O_{\sf B})\\
&= \qty(\hat{\Theta}_{\sf A} \circ \hat{\Theta}_{\sf B} (\openone \otimes O_{\sf A} \otimes \openone)) \cdot \qty(  \hat{\Theta}_{\sf A} \circ \hat{\Theta}_{\sf B} (\openone \otimes \openone \otimes O_{\sf B}))\\
&= \qty(\hat{\Theta}_{\sf A} (\openone \otimes O_{\sf A} \otimes \openone)) \cdot \qty(  \hat{\Theta}_{\sf A} \circ \hat{\Theta}_{\sf B} (\openone \otimes \openone \otimes O_{\sf B}))\\
&= \qty( \Theta_{\sf A} (\openone \otimes O_{\sf A}) \otimes \openone) \cdot \qty(  \hat{\Theta}_{\sf A} \circ \hat{\Theta}_{\sf B} (\openone \otimes \openone \otimes O_{\sf B})).
\end{aligned}
\end{equation}

With this we can show the following.

\begin{lem}
For an arbitrary region $N_\mathsf{A}$ and $C \in \mathcal{P}_\mathsf{A}(N_\mathsf{A}) \otimes \mathcal{P}_\mathsf{B}(K_\mathsf{B}^\perp)$ we have that 
\begin{equation}
    \begin{aligned}
    &(\omega \otimes \sigma_\mathsf{A} \otimes \sigma_\mathsf{B})(\hat{\Theta}_{\mathsf{A}} \circ \hat{\Theta}_{\mathsf{B}} (1\!\!\!\!1 \otimes C))= (\sigma_{\mathsf{A}, \omega, \Theta_{\mathsf{A}}} \otimes \sigma_\mathsf{B})(C).
    \end{aligned}
\end{equation}
\end{lem}

Remark: We do not need to assume any relationship between $\Theta$ and $\hat{\Theta}_{\mathsf{A}} \circ \hat{\Theta}_{\mathsf{B}}$, nor any causal relation between $K_\mathsf{A}$ and $K_\mathsf{B}$ in this lemma.

\begin{proof}
We write $C= \sum_j A_j \otimes B_j$ for $j$ in a finite index set and for $A_j \in \mathcal{P}_\mathsf{A}(N_\mathsf{A})$ and $B_j \in \mathcal{P}_\mathsf{B}(K_\mathsf{B}^\perp)$. For $B \in \mathcal{P}_\mathsf{B}(K_\mathsf{B}^\perp)$ we notice that $\qty(  \hat{\Theta}_{\sf A} \circ \hat{\Theta}_{\sf B} (\openone \otimes \openone \otimes O_{\sf B})) = \openone \otimes \openone \otimes O_{\sf B}$. Then, keeping in mind the observation above, we have
\begin{equation}
    \begin{aligned}
    (\omega \otimes \sigma_\mathsf{A} \otimes \sigma_\mathsf{B})(\hat{\Theta}_{\mathsf{A}} \circ \hat{\Theta}_{\mathsf{B}} (1\!\!\!\!1 \otimes C))&= \sum_j (\omega \otimes \sigma_\mathsf{A} \otimes \sigma_\mathsf{B})(\hat{\Theta}_\mathsf{A}\circ \hat{\Theta}_{\mathsf{B}} (1\!\!\!\!1 \otimes A_j \otimes B_j))\\
    &= \sum_j (\omega \otimes \sigma_\mathsf{A} \otimes \sigma_\mathsf{B})\qty( \Theta_{\sf A} (\openone \otimes A_j) \otimes B_j)\\
    &= \sum_j (\omega \otimes \sigma_\mathsf{A})(\Theta_{\mathsf{A}} (1\!\!\!\!1 \otimes A_j)) \sigma_\mathsf{B}(B_j)\\
    &= \sum_j \sigma_{\mathsf{A}, \omega, \Theta_{\mathsf{A}}}(A_j) \sigma_\mathsf{B}(B_j),
    \end{aligned}
\end{equation}
which shows the desired result.
\end{proof}

\begin{lem}
For $C\in \mathcal{P}_\mathsf{A}({K_\mathsf{A}^\perp}) \otimes \mathcal{P}_\mathsf{B}(N_\mathsf{B})$ where $N_\mathsf{B} \subseteq K_\mathsf{A}^\perp \cap M_\mathsf{B}^+$ precompact, we have that
\begin{equation}
    \begin{aligned}
    &(\omega \otimes \sigma_\mathsf{A} \otimes \sigma_\mathsf{B})(\hat{\Theta}_{\mathsf{A}} \circ \hat{\Theta}_{\mathsf{B}} (1\!\!\!\!1 \otimes C))= (\sigma_{\mathsf{A}} \otimes \sigma_{\mathsf{B}, \omega, \Theta_{\mathsf{B}}})(C).
    \end{aligned}
\end{equation}
\end{lem}

\begin{proof}
We proceed similar {to} the proof of the previous lemma. {Upon writing} $C=\sum_j A_j \otimes B_j$, {the assumption on $N_\mathsf{B}$ allows us to apply} Theorem 2 in~\cite{bostelmann2020impossible}{, which} guarantees that 
\begin{equation}
    \begin{aligned}
    \qty(  \hat{\Theta}_{\sf A} \circ \hat{\Theta}_{\sf B} (\openone \otimes \openone \otimes B_j)) = (\Theta_\mathsf{B}(\openone \otimes B_j))\otimes_2 \openone.
    \end{aligned}
\end{equation}
Hence according to the observation
\begin{equation}
    \begin{aligned}
    (\omega \otimes \sigma_\mathsf{A} \otimes \sigma_\mathsf{B})(\hat{\Theta}_{\mathsf{A}} \circ \hat{\Theta}_{\mathsf{B}} (1\!\!\!\!1 \otimes C))&= \sum_j (\omega \otimes \sigma_\mathsf{A} \otimes \sigma_\mathsf{B})(\hat{\Theta}_\mathsf{A}\circ \hat{\Theta}_{\mathsf{B}} (1\!\!\!\!1 \otimes A_j \otimes B_j))\\
    &= \sum_j (\omega \otimes \sigma_\mathsf{A} \otimes \sigma_\mathsf{B})\qty((\Theta_\mathsf{B}(\openone \otimes B_j))\otimes_2 A_j)\\
    &= \sum_j \sigma_\mathsf{A}(A_j) \; (\omega \otimes \sigma_\mathsf{B})(\Theta_{\mathsf{B}} (1\!\!\!\!1 \otimes B_j))\\
    &= \sum_j  \sigma_{\mathsf{A}}(A_j) \sigma_{\mathsf{B}, \omega, \Theta_{\mathsf{B}}}(B_j),
    \end{aligned}
\end{equation}
which finishes the proof.
\end{proof}

\begin{proof}[Proof of Theorem~\ref{theo_regions_of_separability}.]
The first claim follows immediately by noting that for every $A \in \mathcal{P}((K_\mathsf{A} \cup K_\mathsf{B})^\perp)$ we have that $\Theta(1\!\!\!\!1 \otimes A) = 1\!\!\!\!1 \otimes A$, so on this subalgebra $\sigma' = \sigma$. The remaining two claims follow from the two previous lemmas
\end{proof}

\section{Proof of Lemma~\ref{lem_pure_det}}
\label{sec_appendix_pure_quasi-free_det}

We follow~\cite{petz1990invitation}. Let $\hat{\omega}$ be a pure state on the CCR-$C^*$-algebra of a single local mode. We can find a $2\times2$ matrix $D$ such that $s = D A$.

$D$ has a polar decomposition given by $D=J |D|$, where $J^2 = - \openone$, so $s = J |D| A$. It is a well-established fact that the corresponding state is pure if and only if $|D|=\openone${, see~\cite{petz1990invitation}}. We immediately see that $\det(A) = \det(|D|)^{-1}$, so for pure states we have that $\det(A) = 1$. Conversely, the positivity of the state implies that $\| |D| \| \leq 1$. Using that the norm of $|D|$ is given by the largest eigenvalue and the fact that the determinant is the product of all the eigenvalues, we can deduce that $\det(A){^{-1}}=\det(|D|)=1$ implies that $|D|=\openone$.

\section{Proof of Lemma~\ref{lem_entangled_bipartite_states}}
\label{sec_appendix_proof_lem_entanglement}

Let $W_\mathsf{A}$ and $W_\mathsf{B}$ be the Weyl operators associated to the two probes such that the associated CCR-$C^*$-algebras of the local modes are generated by $\{W_\mathsf{A}(x_1 f_1^\mathsf{A} + x_2 f_2^\mathsf{A}) | \vec{x} \in \mathbb{R}^2\}$ and $\{W_\mathsf{B}(y_1 f_1^\mathsf{B} + y_2 f_2^\mathsf{B}) | \vec{y} \in \mathbb{R}^2\}$ respectively. Let us abuse notation and write $W_\mathsf{A}(\vec{x}) \equiv W_\mathsf{A}(x_1 f_1^\mathsf{A} + x_2 f_2^\mathsf{A})$ and $W_\mathsf{B}(\vec{y})\equiv W_\mathsf{B}(y_1 f_1^\mathsf{B} + y_2 f_2^\mathsf{B})$. Similarly, write $\varphi_{\mathsf{A},\mathsf{B}}{\footnotesize \mqty(\vec{x}\\\vec{y})} \equiv \varphi_\mathsf{A}(x_1 f_1^\mathsf{A} + x_2 f_2^\mathsf{A}) \otimes \openone + \openone \otimes \varphi_\mathsf{B}(y_1 f_1^\mathsf{B} + y_2 f_2^\mathsf{B})$. Let $\omega$ be a quasi-free state on the two modes with covariance matrix $\gamma$ and one-point function $\chi = (\vec{\chi}_\mathsf{A},\vec{\chi}_\mathsf{B})^T \in \mathbb{R}^4${, where we identified $\mathbb{R}^4$ with its dual space via the inner product ``$\cdot$''.} Finally let $\hat{w}$ be the associated quasi-free state on the CCR-$C^*$-theory of the two modes. Then we show that the following are equivalent:
\emph{
\begin{enumerate}
    \item $\hat{\omega}$ is \emph{not} Verch-Werner ppt,
    \item Eq.~\eqref{eq_peres_horodecki} does \emph{not} hold,
    \item $\hat{\omega}$ is entangled,
    \item $\omega$ is entangled.
\end{enumerate}}

\paragraph{$(i) \iff (ii)$:} Let us first look at the Verch-Werner ppt condition. Let $p : \mathbb{R}^2 \to \mathbb{R}^2$ be an $\mathbb{R}$-linear, involutive operator with matrix representation also denoted by $p$ such that $p^T s p= -s$. It gives rise to a $\mathbb{C}$-linear map $\Pi$ from the $\mathsf{B}$-CCR-$C^*$-algebra into itself defined on the Weyl generators as $\Pi (W_\mathsf{B}(\vec{y})) := W_\mathsf{B}(p\vec{y})$. In particular, $\Pi$ preserves the unit. Then the following lemma holds (see also Proposition 3.2 in~\cite{VERCH_2005}).

\begin{lem}
A quasi-free state $\hat{\omega}$ on the CCR-$C^*$-algebra of two local modes fulfills the Verch-Werner ppt condition if and only if
$\underline{\hat{\omega}}(A \otimes B):= \hat{\omega}(A \otimes \Pi (B))$ defines a state {by continuous linear extension}.
\end{lem}

\begin{proof}
As $\underline{\hat{\omega}}$ is clearly linear and normalised for every state $\hat{\omega}$, we prove the equivalence of the ppt property of $\hat{\omega}$ and the positivity of $\underline{\hat{\omega}}$. First we note that $\Pi(W_\mathsf{B}(\vec{y}_1) W_\mathsf{B}(\vec{y}_2)) = e^{-\frac{\mathrm{i}}{2} \vec{y} \cdot s \vec{y}} \Pi(W_\mathsf{B}(\vec{y}_1+\vec{y}_2)) = e^{\frac{\mathrm{i}}{2} p\vec{y} \cdot s p\vec{y}} W_\mathsf{B}(p \vec{y}_1 + p \vec{y}_2) = W_\mathsf{B}(p \vec{y}_2) W_\mathsf{B}(p \vec{y}_2)= \Pi(W_\mathsf{B}(\vec{y}_2)) \Pi( W_\mathsf{B}(\vec{y}_1))$. It is also immediate that $\Pi$ commutes with $\cdot^\dagger$. Then (cf.~Eq.~\eqref{eq_Verch-Werner_ppt})
\begin{equation}
    \begin{aligned}
    &\sum\limits_{\alpha, \beta} \hat{\omega} (A_\beta A_\alpha^\dagger \otimes \Pi(B_\alpha)^\dagger \Pi(B_\beta)) = \sum\limits_{\alpha, \beta} \hat{\omega} (A_\beta A_\alpha^\dagger \otimes \Pi(B_\beta B_\alpha^\dagger))\\
    &=\sum\limits_{\alpha, \beta} \underline{\hat{\omega}}(A_\beta A_\alpha^\dagger \otimes B_\beta B_\alpha^\dagger)= \underline{\hat{\omega}}(X^\dagger X),    
    \end{aligned}
\end{equation}
where $X := \sum\limits_{\alpha} A_\alpha^\dagger \otimes B_\alpha^\dagger$ and where $\alpha, \beta$ each run through the same \emph{finite} index set. The fact that every $X$ in the algebraic tensor product can be written in this form together with a density argument finishes the proof. 
\end{proof}

An example for $p$ is the map derived from $\tilde{p}(y_1 g_1 + y_2 g_2):=y_1 g_1 - y_2 g_2$, i.e, $p= \mathrm{diag}(1, -1)$. In this case, the condition that $\hat{\omega}$ has the Verch-Werner ppt property is equivalent to
\begin{equation}
    \begin{aligned}
    \Lambda^T \gamma \Lambda + \mathrm{i} \qty(s\oplus s) \geq 0,
    \end{aligned}
    \label{eq_lambda_gamma_lambda}
\end{equation}
where $\Lambda = \openone \oplus p =\mathrm{diag}(1,1,1,-1)$, or likewise
\begin{equation}
    \begin{aligned}
    \gamma + \mathrm{i} \qty(s\oplus (-s)) \geq 0,
    \end{aligned}
\end{equation}
which is Eq.~\eqref{eq_peres_horodecki}. Note in particular that $p^T=p$ and $\Lambda^T=\Lambda$.

\paragraph{$(ii) \iff (iii)$ \& $(iv) \implies (ii)$:} Simon showed in~\cite{Simon_2000}, that Eq.~\eqref{eq_peres_horodecki} is equivalent to $\hat{\omega}$ being classically correlated{, which establishes $(ii) \iff (iii)$}. The essence of Simon's proof is that if Eq.~\eqref{eq_peres_horodecki} holds, then the quasi-free state $\omega$ is ``locally related'' to a quasi-free state $\omega_0$ with covariance matrix $\gamma_0$ which fulfills 
\begin{equation}
    \begin{aligned}
    \gamma_0 - \openone \geq 0.
    \end{aligned}
    \label{eq_covariance_separable}
\end{equation}
By ``locally related'' we mean that there exist $*$-automorphisms $\Gamma_\mathsf{A}$ and $\Gamma_\mathsf{B}$ acting on the $\mathsf{A}$ and $\mathsf{B}$ mode respectively such that $\omega_0 := \omega \circ (\Gamma_\mathsf{A} \otimes \Gamma_\mathsf{B})$. In particular $\omega$ is entangled if and only if $\omega_0$ is. A similar statement holds for $\hat{\omega}_0 := \hat{\omega} \circ (\hat{\Gamma}_\mathsf{A} \otimes \hat{\Gamma}_\mathsf{B})$ for $C^*$-automorphisms $\hat{\Gamma}_\mathsf{A}, \hat{\Gamma}_\mathsf{B}$. {Based on this, $(iv) \implies (ii)$ follows from the subsequent lemma, which} shows that $\omega_0$ is classically correlated.

\begin{lem}
Let $\hat{\omega}_0$ be a quasi-free state on the combination of the CCR-$C^*$-theory of two modes with covariance matrix $\gamma_0$ that fulfills Eq.~\eqref{eq_covariance_separable} and one-point function $\chi_0=(\vec{\chi}_\mathsf{A}, \vec{\chi}_\mathsf{B})^T$, then $\omega_0$ is classically correlated (and so is $\hat{\omega}_0$).
\end{lem}

Remark: The proof of this lemma is based on the \emph{Glauber–Sudarshan P-representation}, see their original works~\cite{Glauber1963,Sudarshan1963} and also~Sec.~V.~in~\cite{Englert2002}.

\begin{proof}
Let us define the \emph{coherent} states
\begin{equation}
    \begin{aligned}
    \hat{\eta}^\mathsf{A}_{\vec{\alpha}}(W_\mathsf{A}(\vec{x}))&:= e^{-\frac{1}{4} \vec{x} \cdot \vec{x} + \mathrm{i} \vec{\alpha} \cdot \vec{x} + \mathrm{i} \vec{\chi}_\mathsf{A} \cdot \vec{x}}, \qquad \hat{\eta}^\mathsf{B}_{\vec{\beta}}(W_\mathsf{B}(\vec{y}))&:= e^{-\frac{1}{4} \vec{y} \cdot \vec{y} + \mathrm{i} \vec{\beta} \cdot \vec{y} +  \mathrm{i} \vec{\chi}_\mathsf{B} \cdot \vec{y}},
    \end{aligned}
\end{equation}
which are pure quasi-free states with non-vanishing one-point function.

We assume that Eq.~\eqref{eq_covariance_separable} holds, and discuss the case where $\gamma_0 - \openone$ is positive definite. Then we can define the non-negative Gaussian function $P(\vec{\alpha}, \vec{\beta})$ for the symmetric, positive definite matrix $(\frac{1}{2}\gamma_0 - \frac{1}{2}\openone)^{-1}$ via
\begin{equation}
    \begin{aligned}
    P(\vec{\alpha}, \vec{\beta})&:= \frac{1}{\pi^2 \sqrt{\det(\gamma_0 - \openone)}} e^{- \frac{1}{2} {\tiny \mqty(\vec{\alpha} \\ \vec{\beta})} \cdot  \qty(\frac{1}{2}\gamma_0 -\frac{1}{2} \openone)^{-1} {\tiny\mqty(\vec{\alpha} \\ \vec{\beta})}}. 
    \end{aligned}
\end{equation}
It is then easy to see that for all $\vec{x}, \vec{y} \in \mathbb{R}^2$
\begin{equation}
    \begin{aligned}
    &\int\limits_{\mathbb{R}^2} \int\limits_{\mathbb{R}^2} P(\vec{\alpha}, \vec{\beta}) \hat{\eta}^\mathsf{A}_{\vec{\alpha}}(W_\mathsf{A}(\vec{x})) \; \hat{\eta}^\mathsf{B}_{\vec{\beta}}(W_\mathsf{B}(\vec{y})) \mathrm{d} \vec{\alpha} \;  \mathrm{d} \vec{\beta}\\
    &=\frac{1}{\pi^2 \sqrt{\det(\gamma_0 - \openone)}} e^{\mathrm{i} \chi_0 \cdot {\tiny \mqty(\vec{x} \\ \vec{y})} - \frac{1}{4} {\tiny \mqty(\vec{x} \\ \vec{y})} \cdot {\tiny\mqty(\vec{x} \\ \vec{y})}} \int\limits_{\mathbb{R}^4} e^{- \frac{1}{2} \xi \cdot  \qty(\frac{1}{2}\gamma_0 - \frac{1}{2} \openone)^{-1} \xi+ \mathrm{i} \xi \cdot {\tiny \mqty(\vec{x} \\ \vec{y})}} \; \mathrm{d} \xi\\
    &=\frac{1}{\pi^2 \sqrt{\det(\gamma_0 - \openone)}}
    \frac{(2 \pi)^2}{\sqrt{\det((\frac{1}{2}\gamma_0 - \frac{1}{2}\openone)^{-1})}} e^{\mathrm{i} \chi_0 \cdot {\tiny \mqty(\vec{x} \\ \vec{y})}- \frac{1}{4} {\tiny \mqty(\vec{x} \\ \vec{y})} \cdot {\tiny\mqty(\vec{x} \\ \vec{y})}} e^{- \frac{1}{2} {\tiny\mqty(\vec{x} \\ \vec{y})} \cdot  \qty(\frac{1}{2}\gamma_0 - \frac{1}{2}\openone) {\tiny\mqty(\vec{x} \\ \vec{y})}}\\
    &=  e^{\mathrm{i} \chi_0 \cdot {\tiny \mqty(\vec{x} \\ \vec{y})} - \frac{1}{4} {\tiny\mqty(\vec{x} \\ \vec{y})} \cdot  \gamma_0 {\tiny\mqty(\vec{x} \\ \vec{y})}}=\hat{\omega}_0(W_\mathsf{A}(\vec{x}) \otimes W_\mathsf{B}(\vec{y})),
    \end{aligned}
\end{equation}
where we used $\xi := (\vec{\alpha},\vec{\beta})^T$.
Since the integrand is continuous, we can write the integral as a limit of Riemann sums, each of which corresponds to a convex combination (as $P\geq 0$) of coherent states evaluated on a tensor product of Weyl generators. This shows that $\hat{\omega}_0$ is a pointwise limit of convex combinations of products of coherent states and hence classically correlated. 

By using the corresponding  coherent states on the {polynomial} field-$*$-algebra, the statement holds for $\omega_0$ as well. To see this, use that (according to Eq.~\eqref{eq_n-pt-f_from_exp})
\begin{equation}
\begin{aligned}
    \omega_0\qty(\varphi_{\mathsf{A},\mathsf{B}}{\footnotesize \mqty(\vec{x}\\\vec{y})}^n) &= (- \mathrm{i})^n \qty[\dv[n]{a} e^{\mathrm{i} a \chi_0 \cdot {\tiny \mqty(\vec{x} \\ \vec{y})} - a^2 \frac{1}{4} {\tiny\mqty(\vec{x} \\ \vec{y})} \cdot  \gamma_0 {\tiny\mqty(\vec{x} \\ \vec{y})}}]_{a=0}\\
    &= (- \mathrm{i})^n \qty[\dv[n]{a} \int\limits_{\mathbb{R}^2} \int\limits_{\mathbb{R}^2} P(\vec{\alpha}, \vec{\beta}) e^{-a^2 \frac{1}{4} {\tiny \mqty(\vec{x} \\ \vec{y})} \cdot {\tiny \mqty(\vec{x} \\ \vec{y})} + \mathrm{i} a \qty({\tiny \mqty(\vec{\alpha} \\ \vec{\beta})}  + {\tiny \mqty(\vec{\chi}_\mathsf{A} \\ \vec{\chi}_\mathsf{B})}) \cdot {\tiny \mqty(\vec{x} \\ \vec{y})}} \mathrm{d} \vec{\alpha} \;  \mathrm{d} \vec{\beta}]_{a=0},
\end{aligned}
\end{equation}
change the order of integration and differentiation and approximate the integral in a similar fashion as before.

The case of merely positive \emph{semi}-definite matrix $\gamma_0 - \openone$, i.e., of rank $< 4$, can be treated similarly by using a lower-dimensional integral (cf.~Sec.~V.~in~\cite{Englert2002}).
\end{proof}

\paragraph{$(ii) \implies (iv)$:} The proof of Lemma~\ref{lem_entangled_bipartite_states} is completed by noting that every classically correlated state $\omega$ on the combination of two modes fulfills Eq.~\eqref{eq_peres_horodecki}. 

To see this let us assume that $\omega$ is the pointwise limit of $\omega_n$ of the form $\omega_n = \sum_j \lambda_j \omega_{\mathsf{A},j} \otimes \omega_{\mathsf{B},j}$, where $j$ runs over some finite index set. Let $\gamma$ and $\gamma_n$ be the covariance matrices of $\omega$ and $\omega_n$ respectively and let $\vec{\chi}_{\mathsf{B},j}$ and $B_j$ be the one-point function and the covariance matrix of (the not necessarily quasi-free) state $\omega_{\mathsf{B},j}$ respectively. Then we note that $p \vec{\chi}_{\mathsf{B},j}$ and $p^T B_jp$ define a quasi-free state $\tilde{\omega}_{\mathsf{B},j}$. Positivity of $\tilde{\omega}_{\mathsf{B},j}$ can be seen as follows: the covariance matrix $B_j$ of \emph{every} (not necessarily quasi-free) state $\omega_\mathsf{B}$ fulfills $B_j + \mathrm{i} s \geq 0$. By explicit computation this is found to be equivalent to $B_j - \mathrm{i} s \geq 0$, which is in turn equivalent to $p^T B_jp + \mathrm{i} s \geq 0$. In particular,  $\tilde{\omega}_n := \sum_j \lambda_j \omega_{\mathsf{A},j} \otimes \tilde{\omega}_{\mathsf{B},j}$ is a state with covariance matrix $\tilde{\gamma}_n$ that hence fulfills $\tilde{\gamma}_n + \mathrm{i} \Omega \geq 0$. We note that $\tilde{\gamma}_n = \Lambda^T \gamma_n \Lambda$ and, since $\omega_n \to \omega$ shows that $\gamma_n \to \gamma$, also $\tilde{\gamma}_n \to \Lambda^T \gamma \Lambda$. Finally, as the limit of a convergent sequence of positive semi-definite matrices is positive semi-definite, Eq.~\eqref{eq_lambda_gamma_lambda} holds, which is equivalent to Eq.~\eqref{eq_peres_horodecki}. 

\newpage
\newcommand{\newblock}{}
\providecommand{\noopsort}[1]{}\providecommand{\singleletter}[1]{#1}%
\end{document}